\documentclass[a4paper,USenglish]{lipics-v2016}
\usepackage{amsmath}
\usepackage{tikz-cd}
\usepackage{thmtools}
\usepackage{eucal}
\usepackage{calc}
\usepackage{adjustbox}
\usepackage{wrapfig}
\usepackage{cutwin}
\usepackage{microtype}
\usepackage{chngcntr}
\usepackage{apptools}
\AtAppendix{\counterwithin{proposition}{section}}

\theoremstyle{plain}
\newtheorem{proposition}[theorem]{Proposition}

\opencutright

\newlength{\strutheight}
\newcommand{\cat}[1]{\mathbf{#1}}
\newcommand{\opc}[1]{#1^\mathbf{op}}
\newcommand{\id}{\mathsf{id}}

\newcommand{\im}{\mathsf{img}}
\renewcommand{\ker}{\mathsf{ker}}
\renewcommand{\dim}{\mathsf{dim}}

\newcommand{\Ps}{\mathscr{P}}
\newcommand{\N}{\mathbb{N}}
\newcommand{\F}{\mathbb{F}}

\newcommand{\eword}{\varepsilon}

\newcommand*{\circled}[1]{\tikz[baseline=(char.base)]{\node[shape=circle,draw,inner sep=1.2pt] (char) {#1};}}
\newcommand*{\smallcircled}[1]{{\tikz[baseline=(X.base)]\node(X)[draw,shape=circle,inner sep=0]{\text{\scriptsize\strut$#1$}};}}

\newcommand{\cata}{\mathsf{cat}}
\newcommand{\deriv}{\mathsf{deriv}}

\newcommand{\rch}{\mathsf{rch}}
\newcommand{\obs}{\mathsf{obs}}

\newcommand{\close}{\mathsf{close}}
\newcommand{\cons}{\mathsf{cons}}

\newcommand{\init}{\mathsf{init}}
\newcommand{\out}{\mathsf{out}}

\newcommand{\lstar}{$\mathtt{L}^{\!\star}$}
\newcommand{\ID}{$\mathtt{ID}$}
\newcommand{\lstart}{$\texttt{L}^{\!\star}$}
\newcommand{\IDt}{$\texttt{ID}$}

\newcommand{\row}{\mathsf{row}}

\newcommand{\sift}{\mathsf{sift}}

\newcommand{\lang}{\mathcal{L}}

\newcommand{\initialAlg}{\mathscr{A}}
\newcommand{\finalCoalg}{\mathscr{Z}}

\newcommand{\inc}{\mathsf{inc}}
\newcommand{\res}{\mathsf{res}}

\newcommand{\rev}{\mathsf{rev}}

\newcommand{\parspace}{\vspace{-2ex}}

\title{CALF: Categorical Automata Learning Framework\footnote{This work was partially supported by the ERC Starting Grant ProFoundNet (grant code 679127).}}

\author{Gerco van Heerdt}
\author{Matteo Sammartino}
\author{Alexandra Silva}
\affil{University College London, United Kingdom\\\texttt{\{gerco.heerdt,m.sammartino,alexandra.silva\}@ucl.ac.uk}}
\authorrunning{G. van Heerdt, M. Sammartino, and A. Silva}
\Copyright{Gerco van Heerdt, Matteo Sammartino, and Alexandra Silva}


\EventEditors{Valentin Goranko and Mads Dam}
\EventNoEds{2}
\EventLongTitle{26th EACSL Annual Conference on Computer Science Logic (CSL 2017)}
\EventShortTitle{CSL 2017}
\EventAcronym{CSL}
\EventYear{2017}
\EventDate{August 20--24, 2017}
\EventLocation{Stockholm, Sweden}
\EventLogo{}
\SeriesVolume{82}
\ArticleNo{29}

\begin{document}

\maketitle

\begin{abstract}
	Automata learning is a technique that has successfully been applied in verification, with the automaton type varying depending on the application domain.
	Adaptations of automata learning algorithms for increasingly complex types of automata have to be developed from scratch because there was no abstract theory offering guidelines.
	This makes it hard to devise such algorithms, and it obscures their correctness proofs.
	We introduce a simple category-theoretic formalism that provides an appropriately abstract foundation for studying automata learning.
	Furthermore, our framework establishes formal relations between algorithms for learning, testing, and minimization.
	We illustrate its generality with two examples: deterministic and weighted automata.
\end{abstract}

\section{Introduction}
Automata learning enables the use of model-based verification methods on black-box systems.
Learning algorithms have been successfully applied to find bugs in implementations of network protocols~\cite{ruiter2015}, to provide defense mechanisms against botnets~\cite{cho2010}, to rejuvenate legacy software~\cite{schuts2016}, and to learn an automaton describing the errors in a program up to a user-defined abstraction~\cite{chapman2015}.
Many learning algorithms were originally designed to learn a deterministic automaton and later, due to the demands of the verification task at hand, extended to other types of automata.
For instance, the popular \lstar{} algorithm, devised by Angluin~\cite{angluin1987}, has been adapted for various other types of automata accepting regular languages~\cite{bollig2009,angluin2015}, as well as more expressive automata, including B\"uchi-style automata~\cite{maler1995,angluin2014}, register automata~\cite{bollig2013,cassel2016}, and nominal automata~\cite{moerman2017}.
As the complexity of these automata increases, the learning algorithms proposed for them tend to become more obscure.
More worryingly, the correctness proofs become involved and harder to verify.

This paper aims to introduce an abstract framework for studying automata learning algorithms in which specific instances, correctness proofs, and optimizations can be derived without much effort.
The framework can also be used to study algorithms such as minimization and testing, showing the close connections between the seemingly different problems.
These connections also enable the transfer of extensions and optimizations among algorithms.

First steps towards a categorical understanding of active learning appeared in~\cite{jacobs2014}.
The abstract definitions provided were based on very concrete data structures used in \lstar{}, which then restricted potential generalization to study other learning algorithms and capture other data structures used by more efficient variations of the \lstar{} algorithm.
Moreover, there was no correctness proof, and optimizations and connections to minimization or testing were not explored before in the abstract setting.
In the present paper, we develop a rigorous \emph{categorical automata learning framework} (CALF) that overcomes these limitations and can be more widely applied. We start by giving a general overview of the paper and its results.

\section{Overview}\label{sec:overview}
In this section we provide an overview of CALF.
We first establish some global notation, after which we introduce the basic ingredients of active automata learning.
Finally, we sketch the structure of CALF and highlight our main results.
The reader is assumed to be familiar with elementary category theory and the theory of regular languages.
\parspace
\subparagraph{Notation.}
The set of words over an alphabet $A$ is written $A^*$; $A^{\le n}$ contains all words of length up to $n \in \N$.
We denote the empty word by $\eword$ and concatenation by $\cdot$ or juxtaposition.
We extend concatenation to sets of words $U, V \subseteq A^*$ using $U \cdot V = \{uv \mid u \in U, v \in V\}$.
Languages $\lang \subseteq A^*$ will often be represented as their characteristic functions $\lang \colon A^* \to 2$.
Given sets $X$ and $Y$, we write $Y^X$ for the set of functions $X \to Y$.
The set $1 = \{*\}$, where $*$ is an arbitrary symbol, is sometimes used to represent elements $x \in X$ of a set $X$ as functions $x \colon 1 \to X$.
We write $|X|$ for the size of a set $X$.
\parspace
\subparagraph{Active Automata Learning.}\label{sec:preautlearn}
\emph{Active} automata learning is a widely used technique to learn an automaton model of a system from observations.
It is based on direct interaction of the learner with an \emph{oracle} that can answer different types of queries about the system.
This is in contrast with \emph{passive} learning, where a fixed set of positive and negative examples is provided and no interaction with the system is possible.

Most active learning algorithms assume the ability to perform \emph{membership queries}, where the oracle is asked whether a certain word belongs to the target language $\lang$ over a finite alphabet $A$.
The algorithms we focus on use this ability to approximate the construction of the minimal DFA for $\lang$.
This construction, due to Nerode~\cite{nerode1958}, can be described as follows.
The states of the minimal DFA accepting $\lang$ are given by the image of the function $t_\lang \colon A^* \to 2^{A^*}$ defined by $t_\lang(u)(v) = \lang(uv)$, assigning to each word the residual language after reading that word.
More precisely, the minimal DFA for $\lang$ is given by the four components
\begin{align*}
	M &
		= \{t_\lang(u) \mid u \in A^*\} &
		m_0 &
		= t_\lang(\eword) \\
	F_M &
		= \{t_\lang(u) \mid u \in A^*, \lang(u) = 1\} &
		\delta_M(t_\lang(u), a) &
		= t_\lang(ua),
\end{align*}
where $M$ is a finite set of states, $\delta_M \colon M \times A \to M$ is the transition function, $F_M \subseteq M$ is the set of accepting states, and $m_0 \in M$ is the initial state.
Note that the transition function is well-defined because of the definition of $t_\lang$ and $M$ is finite because the language is regular.

Though we know $M$ is finite, computing it naively as the image of $t_\lang$ does not terminate, since the domain of $t_\lang$ is the infinite set $A^*$.
Learning algorithms approximate $M$ by instead constructing two finite sets $S, E \subseteq A^*$, which induce the function $\row \colon S \cup S \cdot A \to 2^E$ defined as $\row(u)(v) = \lang(uv)$. The name ``row'' reflects the usual representation of this function as a table with rows indexed by $S \cup S \cdot A$ and columns indexed by $E$.
This table is called an \emph{observation table}, and it induces a \emph{hypothesis} DFA given by
\begin{align*}
	H &
		= \{\row(s) \mid s \in S\} &
		h_0 &
		= \row(\eword) \\
	F_H &
		= \{\row(s) \mid s \in S, \lang(s) = 1\} &
		\delta_H(\row(s), a) &
		= \row(sa).
\end{align*}
It is often required that $\eword$ is an element of both $S$ and $E$, so that the initial and accepting states are well-defined.
For the transition function to be well-defined, the observation table is required to satisfy two properties.
These properties were introduced by Gold~\cite{gold1972}; we use the names given by Angluin~\cite{angluin1987}.
\begin{itemize}
	\item
		\textbf{Closedness} states that each transition actually leads to a state of the hypothesis.
		That is, the table is closed if for all $t \in S \cdot A$ there is an $s \in S$ such that $\row(s) = \row(t)$.
	\item
		\textbf{Consistency} states that there is no ambiguity in determining the transitions.
		That is, the table is consistent if for all $s_1, s_2 \in S$ such that $\row(s_1) = \row(s_2)$ we have $\row(s_1a) = \row(s_2a)$ for any $a \in A$.
\end{itemize}
Gold~\cite{gold1972} also explained how to update $S$ and $E$ to satisfy these properties.
If closedness does not hold, then there exists $sa \in S \cdot A$ such that $\row(s') \ne \row(sa)$ for all $s' \in S$.
This is resolved by adding $sa$ to $S$.
If consistency does not hold, then there are $s_1, s_2 \in S$, $a \in A$, and $e \in E$ such that $\row(s_1) = \row(s_2)$, but $\row(s_1a)(e) \ne \row(s_2a)(e)$.
We add $ae$ to $E$ to distinguish $\row(s_1)$ from $\row(s_2)$.
In both cases, the size of the hypothesis increases.
Since the hypothesis cannot be larger than $M$, the table will eventually be closed and consistent.

The hypothesis is an approximation.
Hence, the language it accepts may not be the target language.
Gold~\cite{gold1972} showed, based on earlier algorithms~\cite{ho1966,arbib1969}, that for any two chains of languages $S_1 \subseteq S_2 \subseteq \cdots$ and $E_1 \subseteq E_2 \subseteq \cdots$, both in the limit equaling $A^*$, the following holds.
There exists an $i \in \N$ such that all hypothesis automata derived from $S_j, E_j$ ($j \geq i$) are isomorphic to $M$.
These hypotheses are built after enforcing closedness and consistency for the sets $S_j$ and $E_j$.
Unfortunately, the convergence point $i$ is not known to the learner.
\parspace
\subparagraph{Arbib and Zeiger.}
To obtain a terminating algorithm, an additional assumption is necessary~\cite{moore1956}.
Arbib and Zeiger~\cite{arbib1969} explained the {algorithm of Ho}~\cite{ho1966} and rephrased it for DFAs.
The algorithm assumes an upper bound $n$ on the size of $M$.
One then simply takes $S = E = A^{\le n-1}$ and obtains a DFA isomorphic to $M$ as the hypothesis.
\parspace
\subparagraph{\IDt{}.}
Angluin~\cite{angluin1981} showed that, for algorithms with just this extra assumption, the number of membership queries has a lower bound that is exponential in the size of $M$.
She then introduced an algorithm called {\ID{}} that makes a stronger assumption.
It assumes a finite set $S$ of words such that each state of $M$ (that does not accept the empty language) is reached by reading a word in $S$.
Initializing $E = \{\eword\}$, the algorithm obtains $M$ up to isomorphism after making the table consistent (for a slightly stronger definition of consistency related to the sink state being kept implicit).
This makes \ID{} polynomial in the sizes of $M$, $S$, and $A$.
\parspace
\subparagraph{\lstart{}.}
Subsequently, Angluin introduced the algorithm \lstar{}~\cite{angluin1987}, which makes yet another assumption: it assumes an oracle that can test any DFA for equivalence with the target language.
In the case of a discrepancy, the oracle will provide a \emph{counterexample}, which is a word in the symmetric difference of the two languages.
\lstar{} adds the given counterexample and its prefixes to $S$.
One can show that this means that the next hypothesis will classify the counterexample correctly and that therefore the size of the hypothesis must have increased.
The algorithm is polynomial in the sizes of $M$ and $A$ and in the length of the longest counterexample.
\parspace
\subparagraph{CALF.}
\begin{wrapfigure}{r}{0.485\textwidth}
	\vspace{-10pt}
	\includegraphics[width=.485\textwidth]{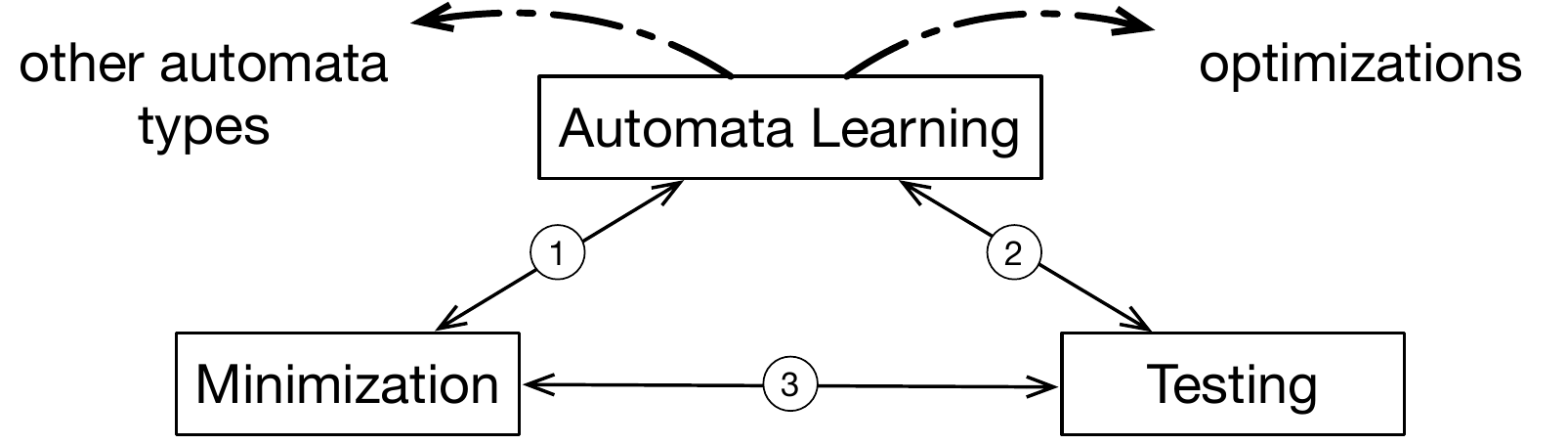}
	\caption{\em Overview of CALF.}\label{fig:calf}
	\vspace{-12pt}
\end{wrapfigure}
The framework we introduce in this paper aims to provide a formal and general account of the concepts introduced above, based on category theory.
Although the focus is on automata learning, the general perspective brought by CALF allows us to unveil deep connections between learning and algorithms for minimization and testing.

The structure of CALF is depicted in \figurename~\ref{fig:calf}.
After introducing the basic elements of the framework in Section~\ref{sec:abstract}, we use the framework in Section~\ref{sec:correctness} to give a new categorical correctness proof, encompassing all the algorithms mentioned above. Then we explore connections with minimization and testing, following the edges of the triangle in \figurename~\ref{fig:calf}:
\begin{enumerate}
		\renewcommand{\labelenumi}{\protect\circled{1}\hspace{-.15cm}}
	\item
		Section~\ref{sec:minimization} shows how our correctness proof also applies to minimization algorithms.
		In particular, we connect consistency fixing to the steps in the classical partition refinement algorithm that merge observationally equivalent states.
		Dually, reachability analysis, which is also needed to get the minimal automaton, is related to fixing closedness defects.
		\renewcommand{\labelenumi}{\protect\circled{2}\hspace{-.15cm}}
	\item
		Section~\ref{sec:testing} extends the correctness proof to an abstract result about testing and we show how it applies to algorithms such as the W-method~\cite{chow1978}, which enable testing against a black-box system.
		\renewcommand{\labelenumi}{\protect\circled{3}\hspace{-.15cm}}
	\item
		From the abstract framework and the observations in \circled{1} and \circled{2}, we also see how certain basic sets used in testing algorithms like the W-method can be obtained using generalized minimization algorithms.
\end{enumerate}
This triangular structure is on its own a contribution of the paper, tying together three important techniques that play a role in automata learning and verification.
Another strong point of CALF is the ability to seamlessly accommodate variants of learning algorithms:
\begin{itemize}
	\item
		algorithms for \emph{other types of automata}, by varying the category on which automata are defined.
		This includes weighted automata, which are defined in the category of vector spaces.
		We present details of this example in Appendix~\ref{sec:inst}.
	\item
		algorithms with \emph{optimizations}, by varying the data structure.
		Section~\ref{sec:opt} discusses \emph{classification trees}~\cite{kearns1994}, which optimize the observation tables of classical learning algorithms.
\end{itemize}
CALF also provides guidelines on how to combine these variants, and on how to obtain corresponding minimization/testing algorithms via the connections described above.
This can lead to more efficient algorithms, and, because of our abstract approach, correctness proofs can be reused.

\section{Abstract Learning}\label{sec:abstract}
In this section, we develop CALF by starting with categorical definitions that capture the basic data structures and definitions used in learning.

\subsection{Observation Tables as Approximated Algebras}
In the automata learning algorithms described in Section~\ref{sec:preautlearn}, observation tables are used to approximate the minimal automaton $M$ for the target language $\lang$.
We start by introducing a notion of approximation of an object in a category, called \emph{wrapper}.
We work in a fixed category $\cat{C}$ unless otherwise indicated.
\begin{definition}[Wrapper]
	A \emph{wrapper for an object $T$} is a pair of morphisms $w = (S \xrightarrow{\sigma} T, T \xrightarrow{\pi} P$); $T$ is called the \emph{target} of $w$.
\end{definition}
\begin{example}\label{ex:twrapper}
	The algorithms explained in Section~\ref{sec:preautlearn} work by producing subsequent approximations of $M$.
	In each approximation, we have two finite sets of words $S, E \subseteq A^*$, yielding a wrapper $w = (S \xrightarrow{\sigma} M, M \xrightarrow{\pi} 2^E)$ in $\cat{Set}$.
	Here $\sigma$ maps $s \in S$ to the state reached by $M$ after reading $s$ while $\pi$ assigns to each state of $M$ the language accepted by that state, but restricted to the words in $E$.
	Although $M$ is unknown, the composition of $\pi$ and $\sigma$ is given by $(\pi \circ \sigma)(s)(e) = \lang(se)$, which can be determined using membership queries.
	This composition is precisely the component $S \to 2^E$ of the observation table $\row\colon S \cup S \cdot A \to 2^E$.
\end{example}
We claim that the other component $S \cdot A \to 2^E$ of $\row$ is an approximated version of the transition function $\delta_M \colon M \times A \to M$, and that it can be obtained via the wrapper $w$.
More generally, we can abstract away from the structure of automata and approximate algebras for a functor $F$, i.e., pairs $(T, f \colon FT \to T)$.
Notice that $(M, \delta_M)$ is an algebra for $(-) \times A$.
\begin{definition}[Approximation of Algebras Along a Wrapper]
	Given an $F$-algebra $(T, f)$ and a wrapper $w = (S \xrightarrow{\sigma} T, T \xrightarrow{\pi} P)$, the \emph{approximation of $f$ along $w$} is $\xi^w_f = FS \xrightarrow{F\sigma} FT \xrightarrow{f} T \xrightarrow{\pi} P$.
	We write $\xi^w$ for $\xi^w_{\id} = \pi \circ \sigma$, and we may leave out superscripts if the relevant wrapper $w$ is clear from the context.
\end{definition}
We remark that our theory does not depend on a specific choice of wrappers and approximations along them.
Observation tables are just an instance.
Other data structures representing approximations of automata can also be captured, as will be shown in Section~\ref{sec:opt}.

Although the algebra $f$ may be unknown---as for instance $f = \delta_M$ is unknown to automata learning algorithms---the approximated version $\xi_f$ can sometimes be recovered via the following simple result, stating that approximations transfer along algebra homomorphisms.
\begin{proposition}\label{prop:shiftwrapper}
	For all endofunctors $F$, $F$-algebras $(U, u)$ and $(V, v)$, $F$-algebra homomorphisms $h \colon (U, u) \to (V, v)$, and morphisms $\alpha \colon S \to U$ and $\beta \colon V \to P$, $\xi^{(h \circ \alpha, \beta)}_v = \xi^{(\alpha, \beta \circ h)}_u$.
\end{proposition}
\begin{proof}
	The $F$-algebra homomorphism $h \colon (U, u) \to (V, v)$ satisfies $v \circ Fh = h \circ u$.
	Then
	\begin{align*}
		\xi_v^{(h \circ \alpha, \beta)} &
			= \beta \circ v \circ F(h \circ \alpha) &
			&
			\text{(definition of $\xi_v^{(h \circ \alpha, \beta)}$)} \\
		&
			= \beta \circ v \circ Fh \circ F\alpha &
			&
			\text{(functoriality of $F$)} \\
		&
			= \beta \circ h \circ u \circ F\alpha &
			&
			\text{($v \circ Fh = h \circ u$)} \\
		&
			= \xi_u^{(\alpha, \beta \circ h)} &
			&
			\text{(definition of $\xi_u^{(\alpha, \beta \circ h)}$)}.
		\qedhere
	\end{align*}
\end{proof}
\begin{example}
	Consider again the wrapper in Example~\ref{ex:twrapper}.
	There $\sigma = \rch \circ \inc$, where $\inc \colon S \to A^*$ is the inclusion of $S$ into $A^*$ and $\rch \colon A^* \to M$ is the full reachability map of $M$.
	If we equip $A^*$ with transitions $\cata \colon A^* \times A \to A^*$ defined as $\cata(u, a) = ua$, then $\rch$ is an algebra homomorphism $(A^*, c) \to (M, \delta)$ for $(-) \times A$.
	We can therefore apply Proposition~\ref{prop:shiftwrapper} to obtain $\xi^{(\rch \circ \inc, \pi)}_\delta = \xi^{(\inc, \pi \circ \rch)}_c$.
	It follows that $\xi_\delta \colon S \times A \to 2^E$ is given by $\xi_\delta(s, a)(e) = \lang(sae)$, which corresponds to the $S \cdot A \to 2^E$ part of the $\row$ function in Section~\ref{sec:preautlearn}.
\end{example}

\subsection{Hypotheses as Factorizations}\label{sec:hyp}
We now formalize the process of deriving a hypothesis automaton from an observation table.
Recall from Section~\ref{sec:preautlearn} that the state space of the hypothesis automaton is precisely the image of the component $S \to 2^E$ of the $\row$ function.
Image factorizations are abstractly captured by the notion of $(\mathcal{E}, \mathcal{M})$ \emph{factorization system} on a category.
We will assume throughout the paper that our category $\cat{C}$ has $(\mathcal{E}, \mathcal{M})$ factorizations.
\begin{definition}[Factorization System]\label{def:fac}
	A \emph{factorization system} is a pair $(\mathcal{E}, \mathcal{M})$ of sets of morphisms such that
	\begin{enumerate}
		\item
			all morphisms $f$ can be decomposed as $f = m \circ e$, with $m \in \mathcal{M}$ and $e \in \mathcal{E}$;
		\item
			\begin{adjustbox}{valign=T,raise=\strutheight,minipage={\linewidth}}
				\begin{wrapfigure}{r}{0.15\textwidth}
					\vspace{-42pt}
					\[
						\begin{tikzcd}[column sep=.6cm,row sep=.45cm,ampersand replacement=\&,cramped]
							U \ar[two heads]{r}{i} \ar{d}[swap]{g} \&
								V \ar{d}{h} \ar[dashed]{dl}[swap]{d} \\
							W \ar[tail]{r}{j} \&
								X
						\end{tikzcd}
					\]
				\end{wrapfigure}
				\strut
				for all commutative squares as on the right, where $i \in \mathcal{E}$ and $j \in \mathcal{M}$, there is a unique diagonal $d$ making the triangles commute;
			\end{adjustbox}\vspace{3pt}
		\item
			both $\mathcal{E}$ and $\mathcal{M}$ are both closed under composition and contain all isomorphisms;
		\item
			every morphism in $\mathcal{E}$ is an epi, and every morphism in $\mathcal{M}$ is a mono.
	\end{enumerate}
\end{definition}
We will use double-headed arrows ($\twoheadrightarrow$) to indicate morphisms in $\mathcal{E}$ and arrows with a tail ($\rightarrowtail$) to indicate morphisms in $\mathcal{M}$.
For instance, in $\cat{Set}$ the pair (surjective functions, injective functions) forms a factorization system, where each function $f \colon X \to Y$ can be decomposed as a surjective map $X \twoheadrightarrow \im(f)$ followed by the inclusion $\im(f) \rightarrowtail Y$.
\renewcommand*\windowpagestuff{%
	\vspace{-8pt}
	\[
		\begin{tikzcd}[column sep=.6cm,row sep=.2cm,ampersand replacement=\&]
			S \ar{r}{\sigma} \ar[bend right,two heads]{dr}{e} \& T \ar{r}{\pi} \&
				P \\
			\&
				H \ar[bend right,tail]{ur}{m}
		\end{tikzcd}
	\]
}
\begin{definition}[Wrapper Minimization]
	\begin{cutout}{0}{\dimexpr\linewidth-3.7cm\relax}{0pt}{3}
		The \emph{minimization of a wrapper $(\sigma, \pi)$} is the $(\mathcal{E}, \mathcal{M})$-factorization of $\pi \circ \sigma$, depicted on the right.
		Notice that $(e, m)$ is a wrapper for $H$.
	\end{cutout}
\end{definition}
In the wrapper of Example~\ref{ex:twrapper}, $H$ is the state space of the hypothesis automaton defined in Section~\ref{sec:preautlearn}.
However, $H$ does not yet have an automaton structure.
In the concrete setting, this can be computed from the $S \cdot A \to 2^E$ part of $\row$ whenever closedness and consistency are satisfied.
We give our abstract characterization of these requirements, which generalize the definitions due to Jacobs and Silva~\cite{jacobs2014}.
\renewcommand*\windowpagestuff{%
	\vspace{-8pt}
	\begin{equation}\label{eq:cc}
		\begin{tikzcd}[column sep=.7cm,row sep=.7cm,ampersand replacement=\&]
			FS \ar{r}{Fe} \ar[dashed]{d}[swap]{\close^w_f} \ar{dr}{\xi_f^w} \&
				FH \ar[dashed]{d}{\cons_f^w} \\
			H \ar{r}[swap]{m} \&
				P
		\end{tikzcd}
	\end{equation}
}
\begin{definition}[Closedness and Consistency]
	\begin{cutout}{2}{\dimexpr\linewidth-4.4cm\relax}{0pt}{4}
		Let $w = (S \xrightarrow{\sigma} T, T \xrightarrow{\pi} P)$ be a wrapper and $(S \xrightarrow{e} H, H \xrightarrow{m} P)$ its minimization.
		Given an endofunctor $F$ and a morphism $f \colon FT \to T$, we say that $w$ is \emph{$f$-closed} if there is a morphism $\close^w_f \colon FS \to H$ making the left triangle in (\ref{eq:cc}) commute; we say that $w$ is \emph{$f$-consistent} if there is a morphism $\cons^w_f \colon FH \to P$ making the right triangle commute.
	\end{cutout}
	\label{def:clos-cons}
\end{definition}
Closedness and consistency together yield an algebra structure on the hypothesis.
This is not just any algebra, but one that relates the wrapper to its minimization, as we show next.
\renewcommand*\windowpagestuff{%
	\vspace{-10pt}
	\begin{equation}\label{eq:struct}
		\begin{tikzcd}[column sep=.8cm,row sep=.4cm,ampersand replacement=\&]
			FT \ar{d}[swap]{f} \&
				FS \ar{l}[swap]{F\sigma} \ar[two heads]{r}{Fe} \&
				FH \ar[dashed]{d}{\theta^w_f} \\
			T \ar{r}{\pi} \&
				P \&
				H \ar[tail]{l}[swap]{m}
		\end{tikzcd}
	\end{equation}
}
\begin{theorem}\label{thm:dynamical}
	\begin{cutout}{1}{\dimexpr\linewidth-5.6cm\relax}{0pt}{4}
		Let $w = (S \xrightarrow{\sigma} T, T \xrightarrow{\pi} P)$ be a wrapper and $(S \xrightarrow{e} H, H \xrightarrow{m} P)$ its minimization.
		For each $\mathcal{E}$-preserving endofunctor $F$ and each morphism $f \colon FT \to T$, $w$ is both $f$-closed and $f$-consistent if and only if there exists $\theta^w_f \colon FH \to H$ making the diagram on the right commute.
	\end{cutout}
\end{theorem}
\begin{proof}
	Given the morphisms $\close_f$ and $\cons_f$, we let $\theta_f$ be the unique diagonal provided by the factorization system for the commutative square (\ref{eq:cc}); conversely, given $\theta_f$, we take $\close_f = \theta_f \circ Fe$ and $\cons_f = m \circ \theta_f$.
\end{proof}
To understand abstract closedness and consistency properties in $\cat{Set}$, and to relate them to the concrete definitions given in Section~\ref{sec:preautlearn}, we give the following general result.
\begin{proposition}\label{prop:closedcons}
	Consider functions $f$, $g$, and $h$ as in the diagram below, with $f$ surjective.
	\begin{enumerate}
		\item
			\begin{adjustbox}{valign=T,raise=\strutheight,minipage={\linewidth}}
				\begin{wrapfigure}{r}{0.2\textwidth}
					\vspace{-26pt}
					\[
						\hspace{-8pt}\begin{tikzcd}[column sep=.8cm,row sep=.3cm,ampersand replacement=\&]
							U \ar[two heads,pos=.4]{r}{f} \ar{rd}{g} \ar[dashed]{d}[swap,pos=.25]{i} \&
								V \ar[dashed,pos=.25]{d}{j} \\
							W \ar[pos=.4]{r}{h} \&
								X
						\end{tikzcd}
					\]
				\end{wrapfigure}\strut
				A function $i \colon U \to W$ making the left triangle in the diagram commute exists if and only if for each $u \in U$ there is a $w \in W$ such that $h(w) = g(u)$.
			\end{adjustbox}\vspace{3pt}
		\item
			A function $j \colon V \to X$ making the right triangle commute exists if and only if for all $u_1, u_2 \in U$ such that $f(u_1) = f(u_2)$ we have $g(u_1) = g(u_2)$.
	\end{enumerate}
\end{proposition}
\begin{proof}
	\begin{enumerate}
		\item
			If such an $i$ exists, then for each $u \in U$ we have $h(i(u)) = g(u)$.
			Conversely, define $i(u)$ to be any $w \in W$ such that $h(w) = g(u)$, which exists by assumption.
			Then $(h \circ i)(u) = h(w) = g(u)$.
		\item
			If such a $j$ exists, then whenever $f(u_1) = f(u_2)$ we have $(j \circ f)(u_1) = (j \circ f)(u_2)$ and therefore $g(u_1) = g(u_2)$.
			Conversely, define $j(f(u))$ to be $g(u)$, using that $f$ is surjective.
			We only need to check that this is well-defined; i.e., that whenever $f(u_1) = f(u_2)$ we also have $g(u_1) = g(u_2)$.
			This is precisely the assumption.
			\qedhere
	\end{enumerate}
\end{proof}

For the wrapper in Example~\ref{ex:twrapper} and $F = (-) \times A$, $\delta_M$-closedness and $\delta_M$-consistency coincide respectively with closedness and consistency as defined in Section~\ref{sec:preautlearn}.
If these are satisfied, then the function $\theta_{\delta_M}$ is precisely the transition function $\delta_H$ defined there. Notice that all endofunctors on $\cat{Set}$ preserve surjections.

Interestingly, closedness and consistency also tell us when initial and accepting states of $H$ can be derived from the observation table. For initial states, we note that the initial state $m_0 \in M$ can be seen as a function $m_0 \colon 1 \to M$.
As the set $1$ is the constant functor $1$ applied to $M$, this initial state gives rise to another closedness property, which states that there must be an $s \in S$ such that $\xi(s)(e) = \lang(e)$ for all $e \in E$.
Note that $m_0$-consistency trivially holds because of the constant functor involved.
When $m_0$-closedness is satisfied, $\close_{m_0} \colon 1 \to H$ is an initial state map.
For instance, in \lstar{} this property is always satisfied because $\eword \in S$.

For accepting states, one would expect a similar property regarding the set $F_M \subseteq M$, which can be represented by a function $F_M \colon M \to 2$.
However, this is a coalgebra (for the constant functor $2$) rather than an algebra.
Fortunately, a wrapper $(S \xrightarrow{\sigma} T, T \xrightarrow{\pi} P)$ in $\cat{C}$ gives a wrapper $(\pi, \sigma)$ in $\opc{\cat{C}}$, so a morphism $f \colon T \to FT$ in $\cat{C}$ yields the approximation $\xi^{(\pi, \sigma)}_f = S \xrightarrow{\sigma} T \xrightarrow{f} FT \xrightarrow{F\pi} FP$.
In particular, for $F_M \colon M \to 2$, this leads to a consistency\footnote{%
	Technically, this should be called coclosedness, as it is closedness in the category $\opc{\cat{C}}$.
	We choose to overload consistency so as not to obscure the terminology.
} property stating that for all $s_1, s_2 \in S$ such that $\xi(s_1) = \xi(s_2)$ we must have $\lang(s_1) = \lang(s_2)$.
The \lstar{} algorithm ensures this by having $\eword \in E$, using that $\lang(s) = \xi(s)(\eword)$ for any $s \in S$.

\section{A General Correctness Theorem}\label{sec:correctness}
In this section we work towards a general correctness theorem that is completed in Section~\ref{sec:automata}, where we introduce an abstract notion of automaton.
We then show in Section~\ref{sec:applications} how the theorem applies to the \ID{} algorithm, to the algorithm by Arbib and Zeiger, and to \lstar{}.

First, we need the following result.
\renewcommand*\windowpagestuff{%
	\vspace{-12pt}
	\[
		\hspace{-4pt}\begin{tikzcd}[row sep=.4cm,column sep=1.2cm,ampersand replacement=\&]
			U \ar[two heads]{r}{e} \ar{dd}[swap]{\id} \&
				I \ar{d}{m} \ar[dashed]{ddl}[swap]{d} \\
				\&
				V \ar{d}{g} \\
			U \ar[tail]{r}{g \circ f} \&
				W
		\end{tikzcd}
	\]
}
\begin{proposition}\label{prop:uncomp}
	If $f \colon U \to V$ and $g \colon V \to W$ are such that $g \circ f \in \mathcal{M}$, then $f \in \mathcal{M}$.
\end{proposition}
\begin{proof}
	\begin{cutout}{0}{\dimexpr\linewidth-3.4cm\relax}{0pt}{5}
		Factorize $f = m \circ e$, with $e \in \mathcal{E}$ and $m \in \mathcal{M}$, and consider the unique diagonal $d$ obtained in the commutative square on the right.
		We see that $d \circ e = \id$.
		Then $e \circ d \circ e = e$, and therefore $e \circ d = \id$ because $e$ is epic.
		Thus, $e$ is an isomorphism.
		Since $\mathcal{M}$ is closed under composition with isomorphisms, we conclude that $f \in \mathcal{M}$.
	\end{cutout}
	\vspace{-14pt}
\end{proof}

\begin{wrapfigure}{r}{0.42\textwidth}
	\vspace{-24pt}
	\begin{align}\label{eq:phipsi}
		\hspace{-8pt}\begin{gathered}
			\begin{tikzcd}[ampersand replacement=\&]
				S \ar[two heads]{r}{\sigma} \ar[two heads]{d}[swap]{e} \&
					T \ar{d}{\pi} \ar[dashed]{dl}[swap]{\phi} \\
				H \ar[tail]{r}{m} \&
					P
			\end{tikzcd}
		\end{gathered} &
			&
			\begin{gathered}
				\begin{tikzcd}[ampersand replacement=\&]
					S \ar[two heads]{r}{e} \ar{d}[swap]{\sigma} \&
						H \ar[tail]{d}{m} \ar[dashed]{dl}[swap]{\psi} \\
					T \ar[tail]{r}{\pi} \&
						P
				\end{tikzcd}
			\end{gathered}
	\end{align}
	\vspace{-20pt}
\end{wrapfigure}
The key observation for the correctness theorem is the following.
Let $w = (S \xrightarrow{\sigma} T, T \xrightarrow{\pi} P)$ be a wrapper with minimization $(S \xrightarrow{e} H, H \xrightarrow{m} P)$.
If $\sigma \in \mathcal{E}$, then the factorization system gives us a unique diagonal $\phi$ in the left square of (\ref{eq:phipsi}), which by (the dual of) Proposition~\ref{prop:uncomp} satisfies $\phi \in \mathcal{E}$.
Similarly, if $\pi \in \mathcal{M}$, we have $\psi$ in the right square of (\ref{eq:phipsi}), with $\psi \in \mathcal{M}$.
Composing the two diagrams and using again the unique diagonal property, one sees that $\phi$ and $\psi$ must be mutually inverse.
We can conclude that $H$ and $T$ are isomorphic.
Now, if $w$ is a wrapper produced by a learning algorithm and $T$ is the state space of the target minimal automaton, as in Example~\ref{ex:twrapper}, our reasoning hints at a correctness criterion: $\sigma \in \mathcal{E}$ and $\pi \in \mathcal{M}$ upon termination ensure that $H$ is isomorphic to $T$.
Of course, the criterion will have to guarantee that the automata, not just the state spaces, are isomorphic.

We first show that the argument above lifts to $F$-algebras $f \colon FT \to T$, for an arbitrary endofunctor $F \colon \cat{C} \to \cat{C}$ preserving $\mathcal{E}$.
\begin{lemma}\label{lem:clco}
	For a wrapper $w = (\sigma, \pi)$ and an $F$-algebra $f$, if $\sigma \in \mathcal{E}$, then $w$ is $f$-closed; if $\pi \in \mathcal{M}$, then $w$ is $f$-consistent.
\end{lemma}
\begin{proof}
	If $\sigma \in \mathcal{E}$, then let $\phi$ be as in (\ref{eq:phipsi}) and define $\close_f = \phi \circ f \circ F\sigma$; if $\pi \in \mathcal{M}$, then let $\psi$ be as in (\ref{eq:phipsi}) and define $\cons_f = \pi \circ f \circ F\psi$.
\end{proof}
\renewcommand*\windowpagestuff{%
	\vspace{-.5cm}
	\[
		\hspace{.2cm}
		\begin{gathered}
			\begin{tikzcd}[column sep=.9cm,row sep=.45cm,ampersand replacement=\&]
				FT \ar{r}{f} \&
					T \ar{r}{\phi} \ar{dr}[pos=.35,swap]{\pi} \&
					H \ar[tail]{d}{m} \\
				FS \ar{rr}[pos=.4]{\xi_f} \ar[two heads]{dd}[swap]{F\sigma} \ar{ddr}[swap,pos=.3]{Fe} \ar[bend left=8]{ddrr}[pos=.75]{\close_f} \ar[two heads]{u}[pos=.3]{F\sigma} \ar[phantom,bend left=4,pos=.4]{ru}{\smallcircled{1}} \ar[phantom,bend right=20,pos=.94]{rru}{\smallcircled{2}} \&
					\&
					P \\
				\\
				FT \ar{r}[swap]{F\phi} \ar[phantom,pos=.1]{rruu}{\smallcircled{3}} \ar[phantom,bend right=5,pos=.45]{rruu}{\smallcircled{4}} \ar[phantom,bend left=5,pos=.7]{rruu}{\smallcircled{5}} \&
					FH \ar{r}[swap]{\theta_f} \&
					H \ar[tail]{uu}[swap]{m}
			\end{tikzcd}
		\end{gathered}
		\quad\hspace{-2pt}
		\begin{array}{ll}
			\smallcircled{1} \text{ definition of $\xi_f$} \\
			\smallcircled{2} \text{ (\ref{eq:phipsi})} \\
			\smallcircled{3} \text{ functoriality, (\ref{eq:phipsi})} \\
			\smallcircled{4} \text{ definition of $\theta_f$} \\
			\smallcircled{5} \text{ closedness}
		\end{array}
	\]
}
\begin{proposition}\label{prop:hyphom}
	For a wrapper $w = (\sigma, \pi)$ and an $F$-algebra $f$, if $\sigma \in \mathcal{E}$ and $w$ is $f$-consistent, then $\phi$ as given in (\ref{eq:phipsi}) is an $F$-algebra homomorphism $(T, f) \to (H, \theta_f)$; if $\pi \in \mathcal{M}$ and $w$ is $f$-closed, then $\psi$ as given in (\ref{eq:phipsi}) is an $F$-algebra homomorphism $(H, \theta_f) \to (T, f)$.
\end{proposition}
\begin{proof}
	\begin{cutout}{0}{\dimexpr\linewidth-8.25cm\relax}{0pt}{7}
		Assume that $\sigma \in \mathcal{E}$ and $w$ is $f$-consistent.
		The proof for the other part is analogous.
		Using Lemma~\ref{lem:clco} and Theorem~\ref{thm:dynamical} we see that $\theta_f$ exists.
		Since $F\sigma$ is epic and $m$ monic, it suffices to show $m \circ \phi \circ f \circ F\sigma = m \circ \theta_f \circ F\phi \circ F\sigma$, which is done on the right.
		\qedhere
	\end{cutout}
\end{proof}
\begin{corollary}\label{cor:hypiso}
	If $\sigma \in \mathcal{E}$ and $\pi \in \mathcal{M}$, then $\phi$ given in (\ref{eq:phipsi}) is an $F$-algebra iso $(T, f) \to (H, \theta_f)$.
\end{corollary}

\subsection{Abstract Automata}\label{sec:automata}
Now we enrich $F$-algebras with initial and final states, obtaining a notion of \emph{automaton} in a category.
Then we give the full correctness theorem for automata.
We fix objects $I$ and $Y$ in $\cat{C}$, which will serve as initial state selector and output of the automaton, respectively.
\renewcommand*\windowpagestuff{%
	\vspace{-10pt}
	\[
		\begin{tikzcd}[column sep=.3cm,row sep=.12cm,ampersand replacement=\&]
			\&
				F Q \ar{ddd}{\delta_Q} \\
			\\
			\\
			\&
				Q \ar{rd}[pos=.4]{\out_Q} \\
			I \ar{ru}[pos=.6]{\init_Q} \&
				\&
				Y
		\end{tikzcd}
	\]
}
\begin{definition}[Automaton]\label{def:aut}
	\begin{cutout}{0}{\dimexpr\linewidth-3.7cm\relax}{0pt}{5}
		An \emph{automaton in $\cat{C}$} is an object $Q$ of $\cat{C}$ equipped with an \emph{initial state map} $\init_Q \colon I \to Q$, an \emph{output map} $\out_Q \colon Q \to Y$, and \emph{dynamics} $\delta_Q \colon F Q \to Q$.
		An \emph{input system} is an automaton without an output map; an \emph{output system} is an automaton without an initial state map.
	\end{cutout}
\end{definition}
Automata form a category, where morphisms $f$ between automata $U$ and $V$ are $F$-algebra homomorphisms that commute with initial state and output maps: $f \circ \init_U = \init_V$ and $\out_V \circ f = \out_U$.
Composition and identities are as in $\cat{C}$.
There are analogous categories of input and output systems.

To recover DAs (DFAs that may not be finite) over an alphabet $A$ as automata, we take $I = 1$, $Y = 2$, and $F = (-) \times A$ in the category $\cat{C} = \cat{Set}$.
An input system is then a DA without a classification of states into accepting and rejecting states; an output system is a DA without a designated initial state.

In Appendix~\ref{sec:inst} we explain how weighted automata can be modeled as automata.
Nominal automata are automata in the category of nominal sets and equivariant functions if we take $I = 1$ with a trivial action, $Y$ any nominal set, and $F = (-) \times A$ for a nominal set $A$.
We do not treat this setting in the present paper, but the \lstar{} adaptation devised by Moerman et al.~\cite{moerman2017} agrees with our abstract definitions.

We assume the existence of an initial input system $\initialAlg$ and a final output system $\finalCoalg$.
These give general notions of \emph{reachability} and \emph{observability}.
\renewcommand*\windowpagestuff{%
	\vspace{-24pt}
	\[
		\begin{tikzcd}[column sep=.92cm,row sep=.28cm,ampersand replacement=\&]
			F \initialAlg \ar[dashed]{r}{F \rch_Q} \ar{dd}[swap]{\delta_\initialAlg} \&
				F Q \ar[dashed]{r}{F \obs_Q} \ar{dd}{\delta_Q} \&
				F \finalCoalg \ar{dd}{\delta_\finalCoalg} \\
			\\
			\initialAlg \ar[dashed]{r}{\rch_Q} \&
				Q \ar[dashed]{r}{\obs_Q} \ar{rd}[pos=.7,swap]{\out_Q} \&
				\finalCoalg \ar{d}{\out_\finalCoalg} \\
			I \ar{u}{\init_\initialAlg} \ar{ru}[pos=.3,swap]{\init_Q} \&
				\&
				Y
		\end{tikzcd}
	\]
}
\begin{definition}[Reachability and Observability]\label{def:reach-obs}
	\begin{cutout}{0}{\dimexpr\linewidth-5.4cm\relax}{0pt}{5}
		The \emph{reachability map} of an automaton $Q$ is the unique input system homomorphism $\rch_Q \colon \initialAlg \to Q$; its \emph{observability map} is the unique output system homomorphism $\obs_Q \colon Q \to \finalCoalg$.
		The automaton $Q$ is \emph{reachable} if $\rch_Q \in \mathcal{E}$; it is \emph{observable} if $\obs_Q \in \mathcal{M}$.
		An automaton is \emph{minimal} if it is both reachable and observable.
	\end{cutout}
\end{definition}
In the DA setting, the set of words $A^*$ forms an initial input system.
Its initial state is the empty word $\eword \colon 1 \to A^*$, and its transition function $\cata \colon A^* \times A \to A^*$ is given by concatenation: $\cata(u, a) = ua$.
This yields the expected definition of the reachability map $\rch_Q \colon A^* \to Q$ for a DA $Q$: $\rch_Q(\eword) = \init_Q(*)$ and $\rch_Q(ua) = \delta_Q(\rch_Q(u), a)$.
The set of languages $2^{A^*}$ in this setting forms a final output system.
Its accepting states $\eword? \colon 2^{A^*} \to 2$ are those languages that contain the empty word, $\eword?(L) = L(\eword)$, and its dynamics $\deriv \colon 2^{A^*} \times A \to 2^{A^*}$ is given by $\deriv(L, a)(v) = L(av)$.
The observability map of a DA $Q$ assigns to each state the language it accepts: $\obs_Q(q)(\eword) = \out_Q(q)$ and $\obs_Q(q)(av) = \obs_Q(\delta_Q(q, a))(v)$.

In general, we define the \emph{language} of an automaton $Q$ to be $\lang_Q = \out_Q \circ \rch_Q \colon \initialAlg \to Y$. For DFAs, this is the usual definition of the accepted language $A^* \to 2$ (alternatively, $\lang_Q = \obs_Q \circ \init_Q \colon 1 \to 2^{A^*}$, which explicitly mentions the initial state).
If for automata $U$ and $V$ there exists an automaton homomorphism $U \to V$, one can prove that $\lang_U = \lang_V$.

We are now ready to give the main result of this section, which extends Proposition~\ref{prop:hyphom}.
Intuitively, it provides conditions for when the hypothesis automaton is actually the target automaton.
Remarkably, unlike Proposition~\ref{prop:hyphom}, it is enough to require just one of $\sigma \in \mathcal{E}$ and $\pi \in \mathcal{M}$, thanks to observability and reachability.
\begin{theorem}\label{thm:hypiso}
	Let $w = (S \xrightarrow{\sigma} Q, Q \xrightarrow{\pi} P)$ be a wrapper for an automaton $Q$ and let $H$ be the hypothesis automaton for $w$.
	If at least one the following is true:
	\begin{enumerate}
		\item\label{hypiso1}
			$Q$ is observable and $w$ is an $\out_Q$-consistent and $\delta_Q$-consistent wrapper such that $\sigma \in \mathcal{E}$;
		\item\label{hypiso2}
			$Q$ is reachable and $w$ is an $\init_Q$-closed and $\delta_Q$-closed wrapper such that $\pi \in \mathcal{M}$;
	\end{enumerate}
	then $H$ and $Q$ are isomorphic automata.
\end{theorem}
\begin{proof}
	We only show point \ref{hypiso1}; the other is analogous.
	Recall that the initial state map of the automaton is an algebra for the constant functor $I$ while the output map is a coalgebra for the constant functor $Y$.
	Hence, we can apply Proposition~\ref{prop:hyphom} (or its dual for the coalgebra) to them and to $\delta_Q$ to find that $\phi \colon Q \to H$ is an automaton homomorphism.
	Then $\obs_Q = \obs_H \circ \phi$ by finality of $\finalCoalg$.
	Since $\obs_Q$ is in $\mathcal{M}$, this means that $\phi \in \mathcal{M}$ (see Proposition~\ref{prop:uncomp}).
	Because $m \in \mathcal{M}$ and $\pi = m \circ \phi$ (\ref{eq:phipsi}), we have $\pi \in \mathcal{M}$.
	Therefore, we can apply Corollary~\ref{cor:hypiso}, again three times, and obtain an automaton isomorphism between $Q$ and $H$.
\end{proof}

\subsection{Applications to Known Learning Algorithms}\label{sec:applications}
Recall that \ID{} assumes a finite set $S \subseteq A^*$ such that each state of the minimal target DFA $M$ is reached by reading one of the words in $S$ when starting from $m_0$.
In the terminology of the previous section, the algorithm takes a wrapper $(S \xrightarrow{\sigma} M, M \xrightarrow{\pi} 2^E)$ as in Example~\ref{ex:twrapper}, with $E$ initialized to $\{\eword\}$, extends $E$ for $\delta$-consistency, and obtains $M$ as the hypothesis.

The correctness of the algorithm can be explained via Theorem~\ref{thm:hypiso}(\ref{hypiso1}) as follows.
The assumption of \ID{} about $S$ is equivalent to $\sigma \in \mathcal{E}$, because $\sigma$ is defined to be the reachability map restricted to $S$.
The $\out$-consistency condition is satisfied by initializing $E = \{\eword\}$, and ensuring $\delta$-consistency is precisely what the algorithm does.
Therefore, by Theorem~\ref{thm:hypiso}(\ref{hypiso1}), the final wrapper yields a hypothesis DFA isomorphic to $M$.

Theorem~\ref{thm:hypiso}(\ref{hypiso2}) suggests a dual to this algorithm, where we assume a finite set $E \subseteq A^*$ such that every pair of different states of $M$ is distinguished by some word in $E$.
Enforcing closedness will lead to a hypothesis DFA isomorphic to $M$.

In a minimal DFA $M$ with at most $n$ states, every state is reached via a word of length at most $n-1$; similarly, every pair of states is distinguished by a word of length up to $n - 1$.
The algorithm of Arbib and Zeiger takes a wrapper $(S \xrightarrow{\sigma} M, M \xrightarrow{\pi} 2^E)$ as in Example~\ref{ex:twrapper}, with $S = E = A^{\le n-1}$, and by Corollary~\ref{cor:hypiso} (applied once for each of $\init_M$, $\delta_M$, and $\out_M$) immediately obtains $M$ up to isomorphism as the hypothesis.
We note that taking this large $E$ is unnecessary, as we could simply apply Theorem~\ref{thm:hypiso}(\ref{hypiso1}), thus reducing the algorithm to \ID{}.

Finally, we consider the \lstar{} algorithm.
Let the wrapper $(S \xrightarrow{\inc} A^* \xrightarrow{\rch} M, M \xrightarrow{\pi} 2^E)$ be as in Example~\ref{ex:twrapper}, and let $H$ be its hypothesis DFA.
We have the following result.

\begin{proposition}\label{prop:lstar}
	If for every prefix $p$ of a word $z \in A^*$ there exists an $s \in S$ such that $\rch_M(s) = \rch_M(p)$, then $\lang_H(z) = \lang_M(z)$.
\end{proposition}
\begin{proof}
	Let $(S \xrightarrow{e} H, H \xrightarrow{m} 2^E)$ be the minimization of the wrapper.
	Furthermore, let $z = a_1 a_2 \ldots a_n$ for $n$ the length of $z$ and each $a_i \in A$.
	For $0 \le i \le n$, define $z_i = a_1 a_2 \ldots a_n$.
	Each prefix of $z$ is $z_i$ for some $i$.
	Assume that for every $i$ there is an $s_i \in S$ such that $\rch_M(s_i) = \rch_M(z_i)$.
	We will show by induction to $n$ that for all $i$, $\xi(s_i) = (m \circ \rch_H)(z_i)$.
	Note that $z_0 = \eword$.
	We have
	\begin{align*}
		\xi(s_0) &
			= (\pi \circ \rch_M)(s_0) &
			&
			\text{(definition of $\xi$)} \\
		&
			= (\pi \circ \rch_M)(\eword) &
			&
			\text{($\rch_M(s_0) = \rch_M(z_0)$)} \\
		&
			= (\pi \circ \init_M)(*) &
			&
			\text{(definition of $\rch_M$)} \\
		&
			= \xi_\init(*) &
			&
			\text{(definition of $\xi_\init$)} \\
		&
			= (m \circ \init_H)(*) &
			&
			\text{(definition of $\init_H$)} \\
		&
			= (m \circ \rch_H)(\eword) &
			&
			\text{(definition of $\rch_H$)}.
	\end{align*}

	Now assume that for a certain $1 \le i < n$ we have $\xi(s_i) = (m \circ \rch_H)(z_i)$.
	This implies $e(s_i) = \rch_H(z_i)$ because $\xi = m \circ e$ and $m$ is injective.
	Then
	\begin{align*}
		\xi(s_{i + 1}) &
			= (\pi \circ \rch_M)(s_{i + 1}) &
			&
			\text{(definition of $\xi$)} \\
		&
			= (\pi \circ \rch_M)(z_{i + 1}) &
			&
			\text{($\rch_M(s_{i + 1}) = \rch_M(z_{i + 1})$)} \\
		&
			= (\pi \circ \rch_M)(z_i a_{i + 1}) &
			&
			\text{($z_{i + 1} = z_i a_{i + 1}$)} \\
		&
			= (\pi \circ \delta_M)(\rch_M(z_i), a_{i + 1}) &
			&
			\text{(definition of $\rch_M$)} \\
		&
			= (\pi \circ \delta_M)(\rch_M(s_i), a_{i + 1}) &
			&
			\text{($\rch_M(s_i) = \rch_M(z_i)$)} \\
		&
			= \xi_\delta(s_i, a_{i + 1}) &
			&
			\text{(definition of $\xi_\delta$)} \\
		&
			= \cons_\delta(e(s_i), a_{i + 1}) &
			&
			\text{(definition of $\cons_\delta$)} \\
		&
			= \cons_\delta(\rch_H(z_i), a_{i + 1}) &
			&
			\text{(induction hypothesis)} \\
		&
			= (m \circ \delta_H)(\rch_H(z_i), a_{i + 1}) &
			&
			\text{(definition of $\delta_H$)} \\
		&
			= (m \circ \rch_H)(z_i a_{i + 1}) &
			&
			\text{(definition of $\rch_H$)} \\
		&
			= (m \circ \rch_H)(z_{i + 1}) &
			&
			\text{($z_{i + 1} = z_i a_{i + 1}$)}.
	\end{align*}
	This concludes the proof that for all $i$, $\xi(s_i) = (m \circ \rch_H)(z_i)$.
	In particular, then, $(m \circ e)(s_n) = \xi(s_n) = (m \circ \rch_H)(z)$.
	Because $m$ is injective, we have $e(s_n) = \rch_H(z)$.
	Therefore,
	\begin{align*}
		\lang_H(z) &
			= (\out_H \circ \rch_H)(z) &
			&
			\text{(definition of $\lang_H$)} \\
		&
			= (\out_H \circ e)(s_n) &
			&
			\text{($e(s_n) = \rch_H(z)$)} \\
		&
			= \xi_\out(s_n) &
			&
			\text{(definition of $\out_H$)} \\
		&
			= (\out_M \circ \rch_M)(s_n) &
			&
			\text{(definition of $\xi_\out$)} \\
		&
			= (\out_M \circ \rch_M)(z) &
			&
			\text{($\rch_M(s_n) = \rch_M(z)$)} \\
		&
			= \lang_M(z) &
			&
			\text{(definition of $\lang_M$)}.
		\qedhere
	\end{align*}
\end{proof}
\begin{corollary}
	If $z \in A^*$ is a counterexample, i.e., $\lang_H(z) \ne \lang_M(z)$, then adding all prefixes of $z$ to $S$ will increase the size of $\im(\rch_M \circ \inc)$.
\end{corollary}
Thus, after \lstar{} has processed at most $|M|$ counterexamples, the conditions for Theorem~\ref{thm:hypiso}(\ref{hypiso1}) are satisfied, which means that the next hypothesis will be isomorphic to $M$.

\section{Minimization}\label{sec:minimization}
In this section we explore the connection $\circled{1}$ in \figurename~\ref{fig:calf} between automata learning and minimization algorithms.
Recall that minimization algorithms typically have two phases: a reachability analysis phase, which removes unreachable states; and a merging phase, where language-equivalent states are merged.
We will rephrase these two phases in the terminology of Section~\ref{sec:preautlearn}, and we will show that Theorem~\ref{thm:hypiso} can be used to explain their correctness.
We fix a DFA $Q$ throughout the section.
\parspace
\subparagraph{Reachability Analysis Phase.}
Let $R$ be the reachable part of $Q$ and $S$ any subset of $R$.
There is a wrapper $w = (S \xrightarrow{\sigma} R, R \xrightarrow{\pi} Q)$, where $\sigma$ and $\pi$ are just the inclusions.
The set $S$ models the state of an algorithm performing a reachability analysis on $Q$.
Since $\sigma$ and $\pi$ are inclusions, the hypothesis for $w$ is the set $S$ itself.

As the inclusion $\pi$ is an automaton homomorphism, by Proposition~\ref{prop:shiftwrapper} we can use the transition function $\delta \colon Q \times A \to Q$ and initial state $q_0 \colon 1 \to Q$ to compute
\begin{align*}
	\xi^{(\sigma, \pi)}_{\delta \colon R \times A \to R} &
		= \xi^{(\pi \circ \sigma, \id)}_{\delta \colon Q \times A \to Q} \colon S \times A \to Q &
		\xi^{(\sigma, \pi)}_{r_0 \colon 1 \to R} &
		= \xi^{(\pi \circ \sigma, \id)}_{q_0 \colon 1 \to Q} = q_0 \colon 1 \to Q.
\end{align*}
The function $\xi_\delta$ simply assigns to each state $s \in S$ and each symbol $a \in A$ the state $\delta_Q(s, a)$.
The wrapper is therefore $\delta_R$-closed if for all $q \in S$ and $a \in A$ we have $\delta_Q(q, a) \in S$; it is $r_0$-closed if $q_0 \in S$.
The obvious algorithm to ensure these closedness properties, namely initializing $S = \{q_0\}$ and adding $\delta_Q(q, a)$ to $S$ while there are $q \in S$ and $a \in A$ such that $\delta_Q(q, a) \not\in S$, is the usual reachability analysis algorithm.
Since $R$ is reachable and $\pi$ injective, Theorem~\ref{thm:hypiso}(\ref{hypiso2}) confirms that this algorithm finds an automaton isomorphic to $R$.

Alternatively, one could let $S$ be a subset of $A^*$ and use the wrapper $(S \xrightarrow{\inc} A^* \xrightarrow{\rch} R, A^* \xrightarrow{\pi} Q)$, where $\inc$ is the inclusion and $\rch$ the reachability map for $R$ (see Definition~\ref{def:reach-obs}).
Starting from $S = \{\eword\}$, the algorithm would add to $S$ words reaching states not yet visited.
This is closer to how automata learning algorithms fix closedness.
\parspace
\subparagraph{State Merging Phase.}
Now we are interested in finding the DFA $O$ that is obtained from $Q$ by merging states that accept the same language.
Formally, this automaton is obtained by factorizing the observability map $\obs_Q$ for $Q$ (see Definition~\ref{def:reach-obs}) as DA homomorphisms $Q \overset{h}{\twoheadrightarrow} O \overset{\obs}{\rightarrowtail} 2^{A^*}$.
Given a finite set $E \subseteq A^*$, these can be made into a wrapper $w = (Q \xrightarrow{h} O, O \xrightarrow{\obs} 2^{A^*} \xrightarrow{\res} 2^E)$, where $\res$ restricts a language to the words in $E$.
Consider $\xi^w \colon Q \to 2^E$ (i.e., the composition of the morphisms in $w$): even though $O$ is not known a priori, this function can be computed by testing for a given state which words in $E$ it accepts.
Because $h$ is an automaton homomorphism, $\xi_{\delta_O}^w$ by Proposition~\ref{prop:shiftwrapper} equals $\xi_{\delta_Q}^{(\id, \xi^w)} = \xi^w \circ \delta_Q$.
Since $h$ is surjective, Theorem~\ref{thm:hypiso}(\ref{hypiso1}) says that we only have to ensure $\delta$-consistency and $\out$-consistency.
One may start with $E = \{\eword\}$ to satisfy the latter.
For $\delta$-consistency, for all $q_1, q_2 \in Q$ such that $\xi(q_1) = \xi(q_2)$ we must have $\xi(\delta(q_1, a)) = \xi(\delta(q_2, a))$ for each $a \in A$.
This can be ensured in the same fashion as for an observation table.

The algorithm we have just described reminds of Moore's reduction procedure~\cite{moore1956}.
However, using a table as data structure is less efficient because many cells may be redundant.
The same redundancy has been observed in the \lstar{} algorithm.
In Section~\ref{sec:opt} we will show how CALF covers optimized versions of these algorithms, but first we discuss yet another application of our framework.

\section{Conformance Testing}\label{sec:testing}
In this section we consider the connections \circled{2} and \circled{3} in \figurename~\ref{fig:calf} between testing and, respectively, automata learning and minimization.
More precisely, we consider an instance of \emph{conformance testing} where a known DFA $U$ is tested for equivalence against a black-box DFA $V$.
One application of this problem is the realization of equivalence queries in the \lstar{} algorithm, where $U$ is the hypothesis DFA and $V$ is the target.

Testing consists in comparing $U$ and $V$ on a finite set of words, approximating their behavior. Following the idea of using wrappers to generalize approximations, we now capture testing using our abstract learning machinery. This will allows us to explain correctness of testing using Theorem~\ref{thm:hypiso}.
First note that given objects $S$ and $P$ and morphisms $\alpha \colon S \to \initialAlg$ and $\omega \colon \finalCoalg \to P$, we can associate wrappers to both the known automaton $U$ and the black-box $V$ by composing with their reachability and observability maps:
\begin{align*}
	w_U &
		= (\sigma_U, \pi_U) &
		\sigma_U &
		= S \xrightarrow{\alpha} \initialAlg \xrightarrow{\rch_U} U &
		\pi_U &
		= U \xrightarrow{\obs_U} \finalCoalg \xrightarrow{\omega} P \\
	w_V &
		= (\sigma_V, \pi_V) &
		\sigma_V &
		= S \xrightarrow{\alpha} \initialAlg \xrightarrow{\rch_V} V &
		\pi_V &
		= V \xrightarrow{\obs_V} \finalCoalg \xrightarrow{\omega} P.
\end{align*}
We now use several approximations in defining the \emph{tests of $U$ against $V$}, covering transition functions, initial states, and final states:
\begin{align}\label{eq:sc}
	\xi^{w_U} &
		= \xi^{w_V} &
		\xi^{w_U}_{\init_U} &
		= \xi^{w_V}_{\init_V} &
		\xi^{w_U}_{\delta_U} &
		= \xi^{w_V}_{\delta_V} &
		\xi^{w_U}_{\out_U} &
		= \xi^{w_V}_{\out_V}.
\end{align}
\begin{example}\label{ex:testing}
	If $U$ and $V$ are DFAs, we can choose wrappers as in Example~\ref{ex:twrapper}, namely by letting $\alpha = \inc \colon S \to A^*$ be an inclusion and $\omega = \res \colon 2^{A^*} \to 2^E$ a restriction.
	The approximations along $w_U$ are then given by $\xi^{w_U}(s)(e) = \lang_U(se)$, $\xi^{w_U}_{\init_U}(e) = \lang_U(e)$, $\xi^{w_U}_{\delta_U}(s,a)(e) = \lang_U(sae)$, and $\xi^{w_U}_{\out_U}(s) = \lang_U(s)$, and analogously for $w_V$.
	Thus, checking (\ref{eq:sc}) here amounts to checking whether the DFAs accept exactly the same words from the set $S \cdot E \cup E \cup S \cdot A \cdot E \cup S$.
\end{example}
Our main result regarding conformance testing captures the properties of the wrappers that need to hold for the above tests to prove that $U$ is equivalent to the black-box $V$.
\begin{theorem}\label{thm:testing}
	Given the above wrappers, suppose $\sigma_U \in \mathcal{E}$, $\pi_U \in \mathcal{M}$, and either $\sigma_V \in \mathcal{E}$ and $V$ is observable or $\pi_V \in \mathcal{M}$ and $V$ is reachable.
	Then $U$ is isomorphic to $V$ if and only if all the equalities (\ref{eq:sc}) hold.
\end{theorem}
\begin{proof}
	Assume first that $U$ is isomorphic to $V$ as witnessed by an isomorphism $\phi \colon U \to V$.
	By initiality, $\phi \circ \rch_U = \rch_V$; by finality, $\obs_V \circ \phi = \obs_U$.
	Then $\obs_U \circ \rch_U = \obs_V \circ \phi \circ \rch_U = \obs_V \circ \rch_V$.
	From this equality the conclusions (\ref{eq:sc}) follow using Proposition~\ref{prop:shiftwrapper}.
	Now assume that the equalities (\ref{eq:sc}) hold.
	From these equations we know that $w_V$ must be $\init$-closed, $\delta$-closed, $\delta$-consistent, and $\out$-consistent, since $w_U$, satisfying $\sigma_U \in \mathcal{E}$ and $\pi_U \in \mathcal{M}$, has these properties by Lemma~\ref{lem:clco}.
	Note that (\ref{eq:sc}) also implies that the hypothesis automata of $w_U$ and $w_V$ coincide.
	Moreover, $\sigma_U \in \mathcal{E}$ and $\pi_U \in \mathcal{M}$ imply that $U$ is minimal because of their definition and Proposition~\ref{prop:uncomp} (and its dual).
	Using this, in combination with the assumptions for $V$, we can now apply Theorem~\ref{thm:hypiso} and conclude that $U$ and $V$ are isomorphic.
\end{proof}
We now comment on the connection between testing and minimization algorithms.
Minimization is a natural choice when looking for a set of words approximating $U$ to be tested against $V$.
Formally, recall from Section~\ref{sec:minimization} that we can use reachability and state merging to find sets $S, E \subseteq A^*$.
Moreover, reachability analysis gives an $\inc \colon S \to A^*$ that makes $\sigma_U$ surjective and state merging gives a $\res \colon 2^{A^*} \to 2^E$ that makes $\pi_U$ injective (recall that $\initialAlg = A^*$ and $\finalCoalg = 2^{A^*}$ in the DFA setting).
These, together with reachability and observability maps, give wrappers for $U$ and $V$.
The condition of Theorem~\ref{thm:testing} on $w_V$ may not hold right away, but these wrappers are convenient starting points for algorithmic techniques, as we now show.
\parspace
\subparagraph{W-method.}
We instantiate the above framework to recover the W-method~\cite{chow1978}.
This algorithm assumes to be given an upper bound $n$ on the number of states of the unknown DFA $V$.
Assume for convenience that $U$ and $V$ are minimal DFAs.
We apply our framework as follows: first, we build the wrapper $w_U = (\sigma_U, \pi_U)$.
We use the minimization algorithms from Section~\ref{sec:minimization} to find $S \xrightarrow{\inc} A^*$ and $2^{A^*} \xrightarrow{\res} 2^E$ with finite $S, E \subseteq A^*$, which yield $\sigma_U = \rch_U \circ \inc \in \mathcal{E}$ and $\pi_U = \res \circ \obs_U \in \mathcal{M}$.
If at this point the equalities (\ref{eq:sc}) do not hold, then we can conclude that $U$ and $V$ accept different languages, and the testing failed.
If we assume they hold, then because $\sigma_U \in \mathcal{E}$ and $\pi_U \in \mathcal{M}$ this means that $|\im(\xi^{w_V})| = |\im(\xi^{w_U})| = |U|$.
The image of $\xi^{w_V} = \pi_V \circ \sigma_V$ is at least as big as the image of $\sigma_V$, so $|\im(\sigma_V)| \ge |U|$.
By assumption we know that the size of $V$ is at most $n$, and hence we update $S$ to $S \cdot A^{\le n - |U|}$, which yields $\sigma_V \in \mathcal{E}$ because $\eword \in S$.
Now we have $\sigma_U \in \mathcal{E}$, $\pi_U \in \mathcal{M}$, and $\sigma_V \in \mathcal{E}$.
Applying Theorem~\ref{thm:testing}, we can find out whether $U$ and $V$ are isomorphic by testing (\ref{eq:sc}) for the updated wrappers.
Instantiating what the equalities in (\ref{eq:sc}) mean (see Example~\ref{ex:testing}), we recover the test sequences generated by the W-method.

\section{Optimized Algorithms}\label{sec:opt}
When fixing a consistency defect in an observation table, we add a column to distinguish two currently equal rows.
Unfortunately, adding a full column implies that a membership query is needed also for each other row.
Kearns and Vazirani~\cite{kearns1994} avoid this redundancy by basing the learning algorithm on a \emph{classification tree} rather than an observation table.
We now show that CALF encompasses this optimized algorithm, which will allow us to derive optimizations for other algorithms.
First we introduce a notion of \emph{classification tree} (also called \emph{discrimination tree}), close to the one by Isberner~\cite{isberner2015}.
\begin{definition}[Classification Tree]
	Given a finite $S \subseteq A^*$, a \emph{classification tree} $\tau$ on $S$ is a labeled binary tree with internal nodes labeled by words from $A^*$ and leaves labeled by subsets of $S$.
\end{definition}
The \emph{language classifier} for $\tau$ is the function $\omega_\tau \colon 2^{A^*} \to \Ps S$ that operates as follows.
Given a language $L \in 2^{A^*}$, starting from the root of the tree: if the current node is a leaf with label $U \in \Ps S$, the function returns $U$; if the current node is an internal node with label $v \in A^*$, we repeat the process from the left subtree if $v \in L$ and from the right subtree otherwise.

In the CALF terminology, each classification tree $\tau$ gives rise to a wrapper $(S \xrightarrow{\inc} A^* \xrightarrow{\rch} M, M \xrightarrow{\obs} 2^{A^*} \xrightarrow{\omega_\tau} \Ps S)$, where $\inc$ is the inclusion and $M$ is the minimal DFA for the language $\lang$ that is to be learned.
We define a function $\sift_\tau \colon A^* \to \Ps S$ as the composition $\omega_\tau \circ \obs_M \circ \rch_M$.
This function \emph{sifts}~\cite{kearns1994} words $u \in A^*$ through the tree by moving on a node with label $v \in A^*$ to the subtree corresponding to the value of $\lang(uv)$.
Applying Proposition~\ref{prop:shiftwrapper}, the approximated initial state map and dynamics of $M$ can be rewritten using this function: $\xi_\init = \sift_\tau \circ \eword \colon 1 \to \Ps S$ and $\xi_\delta = \sift_\tau \circ \cata \circ (\inc \times \id_A) \colon S \times A \to \Ps S$.
Recall that $\eword \colon 1 \to A^*$ is the empty word and $\cata \colon A^* \times A \to A^*$ concatenates its arguments.
The function $\xi_\out \colon S \to 2$ is still the same as for an observation table: $\xi_\out(s) = \lang(s)$.
From these definitions we obtain the following notions of closedness and consistency.
The wrapper is:
\begin{itemize}
	\item
		$\delta$-closed if for each $t \in S \cdot A$ there is an $s \in S$ such that $\sift_\tau(s) = \sift_\tau(t)$;
	\item
		$\delta$-consistent if for all $s_1, s_2 \in S$ such that $\sift_\tau(s_1) = \sift_\tau(s_2)$ we have $\sift_\tau(s_1a) = \sift_\tau(s_2a)$ for any $a \in A$;
	\item
		$\init$-closed if there is an $s \in S$ such that $\sift_\tau(s) = \sift_\tau(\eword)$;
	\item
		$\out$-consistent if for all $s_1, s_2 \in S$ such that $\sift_\tau(s_1) = \sift_\tau(s_2)$ we have $\mathcal{L}(s_1) = \mathcal{L}(s_2)$.
\end{itemize}
Now we can optimize the learning algorithm using Kearns and Vazirani classification trees as follows.
Initially, the tree is just a leaf containing all words in $S$, which may be initialized to $\{\eword\}$.
When an $\out$-consistency or $\delta$-consistency defect is found, we have two words $s_1, s_2 \in S$ such that $\xi(s_1) = \xi(s_2) = U$, for some $U \in \Ps S$.
We also have a word $v \in A^*$ such that $\lang(s_1v) \ne \lang(s_2v),$\footnote{%
	For an $\out$-inconsistency, $v = \eword$; on a $\delta$-inconsistency, there is an $a \in A$ such that $\sift_\tau(s_1a) \ne \sift_\tau(s_2a)$.
	We take the label $u$ of the lowest common ancestor node of those two leaves and define $v = au$.
} and we want to use this word to update the tree $\tau$ to distinguish $\xi(s_1)$ and $\xi(s_2)$.
This is done by replacing the leaf with label $U$ by a node that distinguishes based on the word $v$.
Its left subtree is a leaf containing the words $s \in U$ such that $\lang(sv) = 0$ while the $s \in U$ in its right subtree are such that $\lang(sv) = 1$.
This requires new membership queries, but only one for each word in $S \cup S \cdot A$ that sifts into the leaf with label $U$; the observation table approach needs queries for all elements of $S \cup S \cdot A$ when a column is added.

The intention of having $\Ps S$ as the set of labels is that we maintain the trees in such a way that $\xi \colon S \to \Ps S$ maps each word in $S$ to the unique leaf it is contained in.
As a result, the function $\xi$ can be read directly from the tree.
This means that, when adding a word to $S$, we need to add it also to the leaf of the tree that it sifts into.
Words are added to $S$ when processing a counterexample as in \lstar{} and when fixing $\init$-closedness or $\delta$-closedness.
These closedness defects occur when a word in $\{\eword\} \cup S \cdot A$ sifts into an empty set leaf.
In that case we add the word to $S$.

Classification trees were originally developed for \lstar{}, but we note that they can be used in \ID{} as well.
By the abstract nature of Theorem~\ref{thm:hypiso}, no new correctness proof is necessary.
\parspace
\subparagraph{Transporting Optimizations to Minimization and Testing.}
Using our correspondence between learning and minimization, the above optimization for learning algorithms immediately inspires an optimization for the state merging phase of Section~\ref{sec:minimization}.
The main difference is that we sift states of the automaton $Q$ through the tree, rather than words.
That is, when sifting a state $q$ at a node with label $v \in A^*$, the subtree we move to corresponds to whether $v$ is accepted by $q$. Thus, instead of taking the labels of leaves from $\Ps S$, we take them from $\Ps Q$. The algorithm described above now creates a \emph{splitting tree}~\cite{lee1994,smetsers2016} for $Q$.

Interestingly, in this case one does not actually have to represent the trees to perform the algorithm; tracking only the partitioning of $Q$ induced by the tree is enough.
This is because the classification of next states is contained in the classification of states: $\xi_\delta = \xi \circ \delta$ (using Proposition~\ref{prop:shiftwrapper}).
An inconsistency consists in two states being in the same partition, but for some input $a \in A$ the corresponding next states being in different partitions.
One then splits the partition of the two states into one partition for each of the partitions obtained after reading $a$.
This corresponds to (repeatedly\footnote{%
	The partition splitting algorithm resolves every inconsistency for $a$ within the partition at once.
}) updating the tree to split the leaf.
This algorithm that does not keep track of the tree is precisely Moore's reduction procedure~\cite{moore1956}.

The splitting tree algorithm described above can also be plugged into the W-method as it is described in Section~\ref{sec:testing}.
Note that the discussion about the correctness of the testing algorithm is not affected by this.
The resulting algorithm is closely related to the HSI-method~\cite{luo1995}.

\section{Conclusion}\label{sec:conclusion}
We have presented CALF, a categorical automata learning framework.
CALF covers the basic definitions used in automata learning algorithms and unifies correctness proofs for several of them.
We have shown that these proofs extend also to minimization and testing methods. CALF is general enough to capture optimizations for all of these algorithms and provides an abstract umbrella to transfer results between the three areas.
We illustrated how an optimization known in learning can be transported to minimization and testing.
We have also exploited the categorical nature of the framework by changing the category and deriving algorithms for weighted automata in Appendix~\ref{sec:inst}.
This example shows the versatility of the framework: dynamics of weighted automata are naturally presented as coalgebras, and CALF can accommodate this perspective as well as the algebraic one, which is used traditionally for DFAs and adopted in the main text.
\parspace
\subparagraph{Related Work.}
Most of the present paper is based on the first author's Master's thesis~\cite{heerdt2016}.

Other frameworks have been developed to study automata learning algorithms, e.g.~\cite{balcazar1997,isberner2015}.
However, in both cases the results are much less general than ours and not applicable for example to settings where the automata have additional structure, such as weighted automata (Appendix~\ref{sec:inst}).
Isberner~\cite{isberner2015} hints at the connection between algorithms for minimization and learning.
For example, before introducing classification trees for learning, he explains them for minimization.
Still, he does not provide a formal link between the two algorithms.

The relation between learning and testing was first explored by Berg et al.~\cite{berg2005}.
Their discussion, however, is limited to the case where the black-box DFA has at most as many states as the known DFA.
This is because their correspondence is a stronger one that relates terminated learning algorithms to test sets of conformance testing algorithms, whereas we have provided algorithmic connections.
Our Theorem~\ref{thm:testing}, which allows reasoning about algorithms not making the above assumption, is completely new.

As mentioned in the introduction, a preliminary investigation of generalizing automata learning concepts using category theory was performed by Jacobs and Silva~\cite{jacobs2014}.
We were inspired by their work, but we note that they did not attempt to formulate anything as general as our wrappers; definitions are concrete and dependent on observation tables.
As a result, there are no fully abstract proofs and no instantiations to optimizations such as the classification trees discussed in Section~\ref{sec:opt}.
Jacobs and Silva also investigated weighted automata, as we do in Appendix~\ref{sec:inst}, but only for the \lstar{} algorithm.
We note that Bergadano and Varricchio~\cite{bergadano1996} first adapted \lstar{} and \ID{} for weighted automata.
\parspace
\subparagraph{Future Work.}
Many directions for future research are left open.
An interesting idea inspired by the relation between testing and learning is integrating testing algorithms into \lstar{}.
This could lead to optimizations that are unavailable when those components are kept separate.
Further additions to CALF may include an investigation of the minimality of hypotheses and a characterization of the more elementary steps that are used in the algorithms.
The only fully abstract correctness proofs for concrete algorithms at the moment are for the minimization algorithms and \ID{}, where correctness follows from a condition that can be formulated on an abstract level.
We would like to have an abstract characterization of progress for the other algorithms, which are of a more iterative nature.
Finally, a concrete topic is optimizations for automata with additional structure, such as nondeterministic, weighted, and nominal automata.
To the best of our knowledge, no analogue of classification trees exists for learning these classes.
If such analogues turn out to exist, then the correspondences discussed in Section~\ref{sec:minimization} and Section~\ref{sec:testing} would provide optimized learning and testing algorithms as well.

It is our long-term goal to exploit the practical aspects of the framework.
For more details, see our project website \url{http://calf-project.org}.
\parspace
\subparagraph{Acknowledgements.}
We thank Joshua Moerman and a reviewer for useful comments.

\bibliographystyle{plain}
\bibliography{calf}

\appendix

\section{Linear Weighted Automata}\label{sec:inst}
So far we have only considered DFAs as examples of automata.
In this appendix, we exploit the categorical nature of CALF by changing the base category in order to study linear weighted automata.

\begin{wrapfigure}[5]{r}{0.20\textwidth}
	\vspace{-26pt}
	\[
		\hspace{-12pt}
		\begin{tikzcd}[column sep=.3cm,row sep=.12cm,ampersand replacement=\&]
			\mathbb{F} \ar{rd}[swap,pos=.6]{\init_Q} \&
				\&
				\mathbb{F} \\
			\&
				Q \ar{ddd}{\delta_Q} \ar{ru}[swap,pos=.4]{\out_Q} \\
			\\
			\\
			\&
				Q^A
		\end{tikzcd}
	\]
\end{wrapfigure}
Let $\cat{C}$ be the category $\opc{\cat{Vect}}$, the opposite of the category of vector spaces and linear maps over a field $\F$.
We need the opposite category because the dynamics of our automata will be coalgebras rather than the algebras that are found in Definition~\ref{def:aut}.
We interpret everything in $\cat{Vect}$ rather than $\opc{\cat{Vect}}$.
The automata we consider are for $I = Y = \F$ and $F = (-)^A$, where the latter for a finite set $A$ and a vector space $X$ is given by the set of functions $X^A$ with a pointwise vector space structure.
On linear maps $f \colon W \to X$ we have $f^A \colon W^A \to X^A$ given by $f^A(g)(a) = f(g(a))$.
Here an automaton, called a \emph{linear weighted automaton} (LWA), is a vector space $Q$ with linear maps as in the diagram.

Given a set $U$, let $V(U)$ be the free vector space generated by $U$, the elements of which are formal sums $\sum_{j \in J} v_j \times u_j$ for finite sets $J$, $\{v_j\}_{j \in J} \subseteq \F$, and $\{u_j\}_{j \in J} \subseteq U$.
The operation $V(-)$ is a functor $\cat{Set} \to \cat{Vect}$ that assigns to a function $f \colon U \to W$ for sets $U$ and $W$ the linear map $V(f) \colon V(U) \to V(W)$ given by $V(f)(u) = f(u)$.
The functor $V$ is left adjoint to the forgetful functor $\cat{Vect} \to \cat{Set}$ that assigns to each vector space its underlying set.
Given a vector space $W$ and a function $f \colon U \to W$, the \emph{linearization} of $f$ is the linear map $\overline{f} \colon V(U) \to W$ given by $\overline{f}\left(\sum_{j \in J} v_j \times u_j\right) = \sum_{j \in J} v_j \times f(u_j)$.
Conversely, a linear map with domain $V(U)$ is completely determined by its definition on the elements of $U$.
The adjunction implies that this is a bijective correspondence between functions $U \to W$ and linear maps $V(U) \to W$.
Note that $\F$ is isomorphic to $V(1)$, so the initial state map $\init \colon \F \to Q$ of an LWA $Q$ is essentially an element of $Q$, namely $\init(1)$.
Moreover, any linear map is determined by its value on the basis vectors of its domain, which gives us a finite representation for LWAs with a finite-dimensional state space.
This representation is known as a \emph{weighted automaton}.

The initial input system in this setting is essentially the same as for DAs in $\cat{Set}$, but enriched with a free vector space structure.
\begin{proposition}
	The vector space $V(A^*)$ with
	\begin{align*}
		\init \colon \F \to V(A^*) &
			&
			\init(1) = \eword &
			&
			\delta \colon V(A^*) \to V(A^*)^A &
			&
			\delta(u)(a) = ua
	\end{align*}
	forms an initial input system.
\end{proposition}
\begin{proof}
	Given an input system $Q$ with initial state map $\init_Q \colon V(1) \to Q$ and dynamics $\delta_Q \colon Q \to Q^A$, define a linear map $\rch \colon V(A^*) \to Q$ by $\rch(\eword) = \init_Q(1)$ and $\rch(ua) = \delta_Q(\rch(u))(a)$.
	Note that $\rch$ is an input system homomorphism.
	Any input system homomorphism $h \colon V(A^*) \to Q$ must satisfy $h(\eword) = \init_Q(1)$ $h(ua) = \delta_Q(h(u))(a)$ and is therefore, by induction on the length of words, equal to $\rch$.
\end{proof}
Observability maps are simply as they would be in $\cat{Set}$~\cite[Lemma~7]{jacobs2014}.
Thus, $\obs \colon Q \to \F^{A^*}$ is given by $\obs(q)(\eword) = \out_Q(q)$ and $\obs(q)(au) = \obs(\delta_Q(q)(a))(u)$.
\begin{proposition}
	The vector space $\F^{A^*}$ with
	\begin{align*}
		\out \colon \F^{A^*} \to \F &
			&
			\out(l) = l(\eword) &
			&
			\delta \colon \F^{A^*} \to (\F^{A^*})^A &
			&
			\delta(l)(a)(u) = l(au)
	\end{align*}
	forms a final output system.
\end{proposition}
We use the $(\text{surjective linear maps}, \text{injective linear maps})$ factorization system in $\cat{Vect}$.
The details are similar to those in $\cat{Set}$, but now the image of a linear map $f \colon U \to V$ is a subspace of $V$.
To see that the unique diagonals are the same as in $\cat{Set}$, note the following.
\begin{proposition}\label{prop:linearcompletion}
	For vector spaces $U$, $W$, and $X$, if $f \colon U \to W$ and $g \colon U \to X$ are linear maps and $h \colon X \to W$ is a function with $h \circ f = g$ and $f$ surjective, then $h$ is a linear map.
\end{proposition}
\begin{proof}
	Let $J$ be a finite index set, $\{v_j\}_{j \in J} \subseteq \F$, and $\{w_j\}_{j \in J} \subseteq W$.
	Because $f$ is surjective, we know that for each $j \in J$ there is a $u_j \in U$ such that $f(u_j) = w_j$.
	Then
	\begingroup
	\allowdisplaybreaks
	\begin{align*}
		h\left(\sum_{j \in J} v_j \times w_j\right) &
			= h\left(\sum_{j \in J} v_j \times f(u_j)\right) &
			&
			\text{($f(u_j) = w_j$)} \\
		&
			= h\left(f\left(\sum_{j \in J} v_j \times u_j\right)\right) &
			&
			\text{(linearity of $f$)} \\
		&
			= g\left(\sum_{j \in J} v_j \times u_j\right) &
			&
			\text{($h \circ f = g$)} \\
		&
			= \sum_{j \in J} v_j \times g(u_j) &
			&
			\text{(linearity of $g$)} \\
		&
			= \sum_{j \in J} v_j \times h(f(u_j)) &
			&
			\text{($h \circ f = g$)} \\
		&
			= \sum_{j \in J} v_j \times h(w_j) &
			&
			\text{($f(u_j) = w_j$)}.
		\qedhere
	\end{align*}
	\endgroup
\end{proof}

A language in this setting is a linear map $\lang \colon V(A^*) \to \F$.
One obtains the minimal LWA for the language $\lang$ by taking the image of the observability map $t \colon V(A^*) \to \F^{A^*}$ of the automaton $V(A^*)$ with output map $\lang$.
Note that if we take the initial state map of $\F^{A^*}$ to be $t \circ \init_{V(A^*)}$, then $t$ is an LWA homomorphism.
Its image is an LWA because each category of automata has a factorization system inherited from the base category.
It must be the minimal LWA because by initiality and finality the factorization of $t$ must be $(\rch_M, \obs_M)$.

The dimension of a vector space $U$ is denoted $\dim(U)$.
The kernel of a linear map $f \colon U \to W$ is given by $\ker(f) = \{u \in U \mid f(u) = 0\}$.
It is a standard result that if $U$ is of finite dimension, then
\begin{equation}\label{eq:kerim}
	\dim(U) = \dim(\ker(f)) + \dim(\im(f)).
\end{equation}
Furthermore, if for two vector spaces $U$ and $W$ we have $U \subseteq W$, then $\dim(U) \le \dim(W)$.
If additionally $U \ne W$, then $\dim(U) < \dim(W)$.
For a linear map $f \colon U \to W$, if $U$ is of finite dimension and $\dim(\im(f)) = \dim(U)$, then $\dim(\ker(f)) = 0$ by (\ref{eq:kerim}), implying $f$ is injective.
If $\dim(\im(f)) = \dim(W)$, then $\im(f) = W$, making $f$ surjective.
\parspace
\subparagraph{Learning.}
We recover adaptations of \lstar{} and \ID{} for LWAs as originally developed by Bergadano and Varricchio~\cite{bergadano1996}.
Jacobs and Silva~\cite{jacobs2014} also instantiated their categorical reformulation of \lstar{} to this setting.

In active learning of LWAs, we assume there is a language $\lang \colon V(A^*) \to \F$ such that the state space of the minimal LWA $M$ accepting $\lang$ is of finite dimension.
Given a finite set $E \subseteq A^*$, the function $\res \colon \F^{A^*} \to \F^E$ given by $\res(L)(e) = L(e)$ is a linear map.
To select elements of $V(A^*)$, we may take a finite subset $S \subseteq A^*$ and apply $V$ to the inclusion $\inc \colon S \to A^*$ to obtain a linear map $V(\inc) \colon V(S) \to V(A^*)$ and hence a wrapper
\[
	(\sigma, \pi) = (V(S) \xrightarrow{V(\inc)} V(A^*) \xrightarrow{\rch} M, M \xrightarrow{\obs} \F^{A^*} \xrightarrow{\res} \F^E).
\]
All approximated linear maps are as expected: $\xi \colon V(S) \to \F^E$ is given by $\xi(s)(e) = \lang(se)$, $\xi_\delta \colon V(S) \to (\F^E)^A$ is given by $\xi_\delta(s)(a)(e) = \lang(sae)$, $\xi_\out \colon V(S) \to \F$ is given by $\xi_\out(s) = \lang(s)$, and $\xi_\init \colon \F \to \F^E$ is given by $\xi_\init(1)(e) = \lang(e)$.
Although these maps can be represented in the same kind of observation table used for learning a DFA (except that we have a different output set here), the notions of closedness and consistency are different.
This is because we have to deal with the larger domains of the actual linear maps.
The characterizations of these properties can still be derived from Proposition~\ref{prop:closedcons}, as a result of Proposition~\ref{prop:linearcompletion}.
It follows directly that $\delta$-closedness holds if and only if for all $l \in V(S)$ and $a \in A$ there is an $l' \in V(S)$ such that $\xi(l') = \xi_\delta(l)(a)$.
To simplify this, we have the following result.
\begin{proposition}
	If $J$ is a finite set and $\{v_j\}_{j \in J} \subseteq \F$, $\{s_j\}_{j \in J} \subseteq S$, and $\{l_j\}_{j \in J} \subseteq V(S)$ are such that $\xi(l_j) = \xi_\delta(s_j)(a)$ for all $j \in J$, then $\xi\left(\sum_{j \in J} v_j \times l_j\right) = \xi_\delta\left(\sum_{j \in J} v_j \times s_j\right)(a)$.
\end{proposition}
\begin{proof}
	We have
	\begingroup
	\allowdisplaybreaks
	\begin{align*}
		\xi\left(\sum_{j \in J} v_j \times l_j\right) &
			= \sum_{j \in J} v_j \times \xi(l_j) &
			&
			\text{(linearity of $\xi$)} \\
		&
			= \sum_{j \in J} v_j \times \xi_\delta(s_j)(a) &
			&
			\text{($\xi(l_j) = \xi_\delta(s_j)(a)$)} \\
		&
			= \left(\sum_{j \in J} v_j \times \xi_\delta(s_j)\right)(a) &
			&
			\text{(pointwise vector space structure)} \\
		&
			= \xi_\delta\left(\sum_{j \in J} v_j \times s_j\right)(a) &
			&
			\text{(linearity of $\xi_\delta$)}.
		\qedhere
	\end{align*}
	\endgroup
\end{proof}
Thus, $\delta$-closedness really says that for all $sa \in S \cdot A$ there must be an $l \in V(S)$ such that $\xi(l) = \xi_\delta(s)(a)$.
As for $\delta$-consistency, this property requires that for all $l_1, l_2 \in V(S)$ such that $\xi(l_1) = \xi(l_2)$ we must have $\xi_\delta(l_1)(a) = \xi_\delta(l_2)(a)$ for all $a \in A$.
Using subtraction, this says that for all $l \in V(S)$ with $\xi(l) = \mathbf{0}$ we must have $\xi_\delta(l)(a) = \mathbf{0}$ for all $a \in A$, where $\mathbf{0} \colon E \to \F$ is given by $\mathbf{0}(e) = 0$ for all $e \in E$.

Analogously, $\init$-closedness requires there to be an $l \in V(S)$ such that $\xi(l) = \xi_\init(1)$ while $\out$-consistency states that for all $l \in V(S)$ such that $\xi(l) = \mathbf{0}$ we must have $\lang(l) = 0$.

Closedness can be determined simply by solving systems of linear equations.
Consistency is less trivial.
Define the \emph{transpose} of a table for the language $\lang$ given by $S, E \subseteq A^*$ as the table with row labels $\rev(E)$ and column labels $\rev(S)$ for the language $\rev(\lang)$, where $\rev$ reverses all words in a language.
We will show that consistency of a table can be ensured by ensuring closedness of the transposed table.
This is essentially what Bergadano and Varricchio~\cite{bergadano1996} do.
\begin{proposition}
	If the transpose of a table is $\delta$-closed, then that table is $\delta$-consistent.
\end{proposition}
\begin{proof}
	Suppose there is an $l \in V(S)$ such that $\xi(l) = \mathbf{0}$.
	For all $a \in A$ and $e \in E$ there are by closedness of the transposed table a finite set $J$, $\{v_j\}_{j \in J} \subseteq \F$, and $\{e_j\}_{j \in J} \subseteq E$ such that for every $s \in S$,
	\begin{equation}\label{eq:trcl}
		\xi_\delta(s)(a)(e) = \sum_{j \in J} v_j \times \xi(s)(e_j).
	\end{equation}
	Let $K$ be a finite set and $\{v_k'\}_{k \in K} \subseteq \F$ and $\{s_k\}_{k \in K} \subseteq S$ such that $l = \sum_{k \in K} v_k' \times s_k$.
	Then
	\begingroup
	\allowdisplaybreaks
	\begin{align*}
		\xi_\delta(l)(a)(e) &
			= \xi_\delta\left(\sum_{k \in K} v_k' \times s_k\right)(a)(e) &
			&
			\text{(expanded form of $l$)} \\
		&
			= \left(\sum_{k \in K} v_k' \times \xi_\delta(s_k)\right)(a)(e) &
			&
			\text{(linearity of $\xi_\delta$)} \\
		&
			= \sum_{k \in K} v_k' \times \xi_\delta(s_k)(a)(e) &
			&
			\text{(pointwise vector space structure)} \\
		&
			= \sum_{k \in K} v_k' \times \sum_{j \in J} v_j \times \xi(s_k)(e_j) &
			&
			\text{(\ref{eq:trcl})} \\
		&
			= \sum_{j \in J} v_j \times \sum_{k \in K} v_k' \times \xi(s_k)(e_j) \\
		&
			= \sum_{j \in J} v_j \times \left(\sum_{k \in K} v_k' \times \xi(s_k)\right)(e_j) &
			&
			\text{(pointwise vector space structure)} \\
		&
			= \sum_{j \in J} v_j \times \xi\left(\sum_{k \in K} v_k' \times s_k\right)(e_j) &
			&
			\text{(linearity of $\xi$)} \\
		&
			= \sum_{j \in J} v_j \times \xi(l)(e_j) &
			&
			\text{(expanded form of $l$)} \\
		&
			= \sum_{j \in J} v_j \times 0 &
			&
			\text{($\xi(l) = \mathbf{0}$)} \\
		&
			= 0. &
			&
			\qedhere
	\end{align*}
	\endgroup
\end{proof}
Analogously, one can show that $\init$-closedness of the transposed table implies $\init$-consistency of the original one.

If a closedness defect is found, we have a word $u \in A^*$ such that $(\res \circ \obs_M \circ \rch_M)(u)$ is not a linear combination of any row indexed by words from $S$.
Because we have subtraction and scalar division, this also means that each of the rows indexed by $S$ that is not a linear combination of other rows is also not a linear combination of other rows and $(\res \circ \obs_M \circ \rch_M)(u)$.
Therefore, adding $u$ to $S$ increases the dimension of the hypothesis.
The dimension of the hypothesis cannot exceed the dimension of the target $M$ (see Appendix~\ref{sec:smaller}), which is of finite dimension, so this process must terminate.

Note that the top part of a table in this setting can be seen as a matrix, and that transposing the table transposes this matrix.
By fixing a closedness defect in the transpose of the table, the dimension of its hypothesis increases.
This dimension is precisely the column rank of the original matrix.
Since the column and row ranks of a matrix are equal, the dimension of the hypothesis of the original table must have increased.

The map $\sigma \colon V(S) \to M$ is surjective just if for each $m \in M$ there is an $l \in V(S)$ such that $\rch_M(l) = m$.
Equivalently, the set $\{\rch_M(s) \mid s \in S\}$ needs to span $M$.
Note that this means the set needs to include one of the (finite) bases of $M$.
Thus, we have an adaptation of \ID{} for LWAs that assumes to be given such a finite set $S$ and enforces $\out$-consistency and $\delta$-consistency to obtain, by Theorem~\ref{thm:hypiso}(\ref{hypiso1}), a hypothesis isomorphic to $M$.

As for \lstar{}, Proposition~\ref{prop:lstar} carries over almost immediately.
Due to the similarity, we omit the proof.
\begin{proposition}
	If for every prefix $p$ of a word $z \in A^*$ there exists an $l \in V(S)$ such that $\rch_M(l) = \rch_M(p)$, then $\lang_H(z) = \lang_M(z)$.
\end{proposition}
\begin{corollary}
	If $z \in A^*$ is a counterexample, i.e., $\lang_H(z) \ne \lang_M(z)$, then adding all prefixes of $z$ to $S$ will increase the dimension of $\im(\rch_M \circ \inc)$.
\end{corollary}
Once the dimensions of $\im(\rch_M \circ \inc)$ and $M$ coincide, $\rch_M \circ \inc$ is surjective and we can apply Theorem~\ref{thm:hypiso}(\ref{hypiso1}).
Thus, the number of required counterexamples is bounded by the dimension of $M$.
\parspace
\subparagraph{Reachability Analysis.}
Let $Q$ be a finite-dimensional LWA and $R$ its reachable part.
Given a finite subset $S \subseteq R$, we have a wrapper $(V(S) \xrightarrow{\overline{\sigma}} R, R \xrightarrow{\pi} Q)$ in $\cat{Vect}$, where $\sigma$ and $\pi$ are the inclusions.
This wrapper is $\init$-closed if there is $l \in V(S)$ such that $\xi(l) = \init_Q(1)$, which can be satisfied by initializing $S = \{\eword\}$.
The wrapper is $\delta$-closed if for all $s \in S$ and $a \in A$ there is an $l \in V(S)$ such that $\xi(l) = \delta_Q(\xi(s))(a)$.
If for some $sa$ this is not the case, we add it to $S$.
As in learning, the dimension of the hypothesis increases by doing so, and therefore the process terminates.
Using Theorem~\ref{thm:hypiso}(\ref{hypiso2}), the final hypothesis is isomorphic to $R$.
\parspace
\subparagraph{State Merging.}
Given a finite-dimensional LWA $Q$, we have a factorization of the observability map $\obs_Q$ as in Section~\ref{sec:minimization}---$Q \overset{h}{\twoheadrightarrow} O \overset{\obs}{\rightarrowtail} \F^{A^*}$---yielding an observable equivalent LWA $O$.
Fixing a finite set $E \subseteq A^*$, we have a wrapper $(Q \xrightarrow{h} O, O \xrightarrow{\obs} \F^{A^*} \xrightarrow{\res} \F^E)$.
The linear map $\xi \colon Q \to \F^E$ gives the output of each state on the words in $E$.
The wrapper is $\out$-consistent if for all $q_1, q_2 \in Q$ such that $\xi(q_1) = \xi(q_2)$ we have $\out(q_1) = \out(q_2)$.
If the last equation fails, we may add $\eword$ to $E$ to distinguish $\xi(q_1)$ and $\xi(q_2)$.
The wrapper is $\delta$-consistent if for all $q_1, q_2 \in Q$ and $a \in A$ such that $\xi(q_1) = \xi(q_2)$ we have $\xi(\delta(q_1)(a)) = \xi(\delta(q_2)(a))$.
If this last equation fails on some $e \in E$, we may add $ae$ to $E$ to distinguish $\xi(q_1)$ and $\xi(q_2)$.
As in learning, this decreases the dimension of the kernel of $\xi$, which guarantees termination.
Using Theorem~\ref{thm:hypiso}(\ref{hypiso1}), the final hypothesis is isomorphic to $O$.
\parspace
\subparagraph{Testing.}
Consider a known finite-dimensional LWA $X$ and an unknown finite-dimensional LWA $Z$, both of which are minimal.
Using the minimization algorithms (but with $S$ for the reachability analysis a subset of $A^*$), we can find finite $S, E \subseteq A^*$ such that $\eword \in S$ and the wrapper $w_X = (\sigma_X, \pi_X) = (V(S) \xrightarrow{V(\inc)} V(A^*) \xrightarrow{\rch} X, X \xrightarrow{\obs} \F^{A^*} \xrightarrow{\res} \F^E)$ satisfies $\sigma_X \in \mathcal{E}$ and $\pi_X \in \mathcal{M}$.
Define $w_Z = (\sigma_Z, \pi_Z)$ analogously.
If at this point the equalities
\begin{align}\label{eq:scl}
	\xi^{w_X} &
		= \xi^{w_Z} &
		\xi^{w_X}_{\init_X} &
		= \xi^{w_Z}_{\init_Z} &
		\xi^{w_X}_{\delta_X} &
		= \xi^{w_Z}_{\delta_Z} &
		\xi^{w_X}_{\out_X} &
		= \xi^{w_Z}_{\out_Z}.
\end{align}
given by Theorem~\ref{thm:testing} do not hold, we can conclude that $X$ and $Z$ accept different languages.
Assume these equalities do hold.
By Corollary~\ref{cor:hypiso} this implies $\dim(\im(\xi^{w_Z})) = \dim(U)$.
Assuming a given upperbound $n$ on the dimension of $Z$, updating $S$ to $S \cdot A^{\le n - \dim(X)}$ yields $\sigma_Z \in \mathcal{E}$.
Applying Theorem~\ref{thm:testing}, we can find out whether $X$ and $Z$ are isomorphic by testing (\ref{eq:scl}) for the updated wrappers.
That is, we have an adaptation of the W-method for LWAs.

\section{The Hypothesis is Smaller than the Target}\label{sec:smaller}
\begin{wrapfigure}{r}{.34\textwidth}
	\vspace{-25pt}
	\[
		\hspace{-10pt}\begin{tikzcd}[column sep=.4cm,row sep=.3cm]
			S \ar{rr}{\sigma} \ar[two heads]{rd} &
				&
				T \ar{rr}{\pi} \ar[two heads]{rd} &
				&
				P \\
			&
				J \ar[two heads]{rd} \ar[tail]{ru} &
				&
				K \ar[tail]{ru} \\
			&
				&
				H \ar[tail]{ru}
		\end{tikzcd}
	\]
	\vspace{-30pt}
\end{wrapfigure}
Given a wrapper $(S \xrightarrow{\sigma} T, T \xrightarrow{\pi} P)$ in an arbitrary category $\cat{C}$ with an $(\mathcal{E}, \mathcal{M})$ factorization system, the hypothesis can be defined using three successive factorizations, as shown on the right.
We now explain for the discussed concrete settings why the hypothesis is always ``smaller'' than the target, for the appropriate notion of size in each category.
\parspace
\subparagraph{In the Category of Sets.}
We know that if $T$ is a finite set, then $J$, being a subset of $T$, is a finite set with $|J| \le |T|$.
Because there is a surjective function $J \to H$, we conclude that $H$ is finite and $|H| \le |J| \le |T|$.
\parspace
\subparagraph{In the Category of Vector Spaces.}
Assume $T$ is of finite dimension.
Since $J$ is a subspace of $T$, we have that $J$ is a finite-dimensional vector space with $\dim(J) \le \dim(T)$.
Note that there exists a surjective linear map $x \colon J \to H$, so $H$ is finite-dimensional and $\dim(\im(x)) = \dim(H)$.
It then follows from (\ref{eq:kerim}) in Appendix~\ref{sec:inst} that $\dim(\ker(x)) = \dim(J) - \dim(H)$, so we must have $\dim(H) \le \dim(J) \le \dim(T)$.

\end{document}